\documentclass{amsart}
\usepackage{hyperref}
\usepackage{graphicx}
\usepackage{curves}

\newcommand*{\mailto}[1]{\href{mailto:#1}{\nolinkurl{#1}}}

\newtheorem{theorem}{Theorem}[section]
\newtheorem{lemma}[theorem]{Lemma}
\newtheorem{corollary}[theorem]{Corollary}
\newtheorem{remark}[theorem]{Remark}
\newtheorem{hypothesis}[theorem]{Hypothesis}

\newcommand{\R}{\mathbb{R}}

\newcommand{\C}{\mathbb{C}}

\newcommand{\nn}{\nonumber}
\newcommand{\be}{\begin{equation}}
\newcommand{\ee}{\end{equation}}
\newcommand{\bea}{\begin{eqnarray}}
\newcommand{\eea}{\end{eqnarray}}

\newcommand{\ol}{\overline}
\newcommand{\ti}{\tilde}
\newcommand{\wti}{\widetilde}
\newcommand{\wha}{\widehat}

\newcommand{\abs}[1]{\lvert#1 \rvert}

\newcommand{\id}{\mathbb{I}}
\newcommand{\I}{\mathrm{i}}
\newcommand{\E}{\mathrm{e}}
\newcommand{\ind}{\mathop{\mathrm{ind}}}
\newcommand{\re}{\mathop{\mathrm{Re}}}
\newcommand{\im}{\mathop{\mathrm{Im}}}
\newcommand{\sech}{\mathop{\mathrm{sech}}}
\DeclareMathOperator{\res}{Res}
\newcommand{\db}{\mathfrak{D}}

\def\Xint#1{\mathchoice
   {\XXint\displaystyle\textstyle{#1}}%
   {\XXint\textstyle\scriptstyle{#1}}%
   {\XXint\scriptstyle\scriptscriptstyle{#1}}%
   {\XXint\scriptscriptstyle\scriptscriptstyle{#1}}%
   \!\int}
\def\XXint#1#2#3{{\setbox0=\hbox{$#1{#2#3}{\int}$}
     \vcenter{\hbox{$#2#3$}}\kern-.5\wd0}}
\def\dashint{\Xint-}

\newcommand{\eps}{\varepsilon}

\newcommand{\gam}{\gamma}

\numberwithin{equation}{section}

\newcommand{\sigI}{\begin{pmatrix} 0 & 1 \\ 1 & 0 \end{pmatrix}}
\newcommand{\rI}{\begin{pmatrix}  1 & 1 \end{pmatrix}}
\newcommand{\rN}{\begin{pmatrix}  0 & 0 \end{pmatrix}}

\begin{document}

\title[Long-Time Asymptotics for the KdV Equation]{Long-Time Asymptotics for the Korteweg--de Vries Equation
via Nonlinear Steepest Descent}

\author[K. Grunert]{Katrin Grunert}
\address{Faculty of Mathematics\\ Nordbergstrasse 15\\ 1090 Wien\\ Austria}
\email{\mailto{katrin.grunert@univie.ac.at}}
\urladdr{\url{http://www.mat.univie.ac.at/~grunert/}}

\author[G. Teschl]{Gerald Teschl}
\address{Faculty of Mathematics\\
Nordbergstrasse 15\\ 1090 Wien\\ Austria\\ and International Erwin Schr\"odinger
Institute for Mathematical Physics\\ Boltzmanngasse 9\\ 1090 Wien\\ Austria}
\email{\mailto{Gerald.Teschl@univie.ac.at}}
\urladdr{\url{http://www.mat.univie.ac.at/~gerald/}}

\thanks{Research supported by the Austrian Science Fund (FWF) under Grant No.\ Y330.}
\thanks{Math. Phys. Anal. Geom. {\bf 12}, 287--324 (2009)}

\keywords{Riemann--Hilbert problem, KdV equation, solitons}
\subjclass[2000]{Primary 37K40, 35Q53; Secondary 37K45, 35Q15}

\begin{abstract}
We apply the method of nonlinear steepest descent to compute the long-time
asymptotics of the Korteweg--de Vries equation for decaying initial data in the soliton and similarity region.
This paper can be viewed as an expository introduction to this method.
\end{abstract}

\maketitle

\section{Introduction}

One of the most famous examples of completely integrable wave equations is the Korteweg--de Vries (KdV) equation
\be\label{kdv}
q_t(x,t)=6q(x,t)q_x(x,t)-q_{xxx}(x,t), \quad (x,t)\in\R\times\R,
\ee
where, as usual, the subscripts denote the differentiation with 
respect to the corresponding variables.

Following the seminal work of Gardner, Green, Kruskal, and Miura \cite{ggkm}, one can use the inverse scattering transform
to establish existence and uniqueness of (real-valued) classical solutions for the corresponding initial value problem
with rapidly decaying initial conditions. We refer to, for instance,
the monographs by Marchenko \cite{mar} or Eckhaus and Van Harten \cite{evh}. Our concern here
are the long-time asymptotics of such solutions. The classical result is that an arbitrary short-range solution of the
above type will eventually split into a number of solitons travelling to the right plus a decaying radiation part
travelling to the left, as illustrated in Figure~\ref{fig1}.
\begin{figure}
\includegraphics[width=8cm]{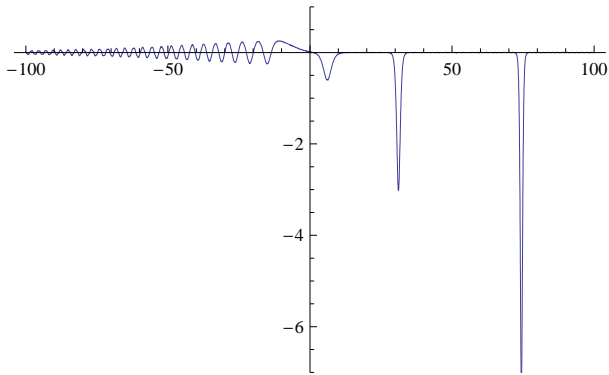}
\caption{Numerically computed solution $q(x,t)$ of the KdV equation at time $t=5$, with initial
condition $q(x,0)=\sech(x+3)-5\sech(x-1)$.} \label{fig1}
\end{figure}
The first numerical evidence for such a behaviour was found by Zabusky and Kruskal \cite{zakr}. The first mathematical
results were given by Ablowitz and Newell \cite{an}, Manakov \cite{ma}, and \v{S}abat \cite{sh}. First rigorous results for the
KdV equation were proved by \v{S}abat \cite{sh} and Tanaka \cite{ta2}
(see also Eckhaus and Schuur \cite{es}, where more detailed error bounds are given). Precise asymptotics for the radiation part were first formally
derived by Zakharov and Manakov \cite{zama}, by Ablowitz and Segur \cite{as}, \cite{as2}, by Buslaev \cite{bu} (see also \cite{bb}),
and later on rigorously justified and extended to all orders by Buslaev and Sukhanov \cite{bs}. A detailed rigorous proof
(not requiring any a priori information on the asymptotic form of the solution) was given by Deift and Zhou \cite{dz} based on earlier
work of Manakov \cite{ma} and Its \cite{its1} (see also \cite{its2}, \cite{its3}, \cite{ip}).
For further information on the history of this problem we refer to the survey by Deift, Its, and Zhou \cite{diz}.

To describe the asymptotics in more detail, we recall the well-known fact (see e.g.\ \cite{dt}, \cite{mar}) that $q(x,t)$ is
uniquely determined by the (right) scattering data of the associated Schr\"odinger operator 
\be
H(t)=-\frac{d^2}{dx^2}+q(x,t).
\ee
The scattering data consist of the (right) reflection coefficient $R(k,t)$, a finite number of ($t$ independent) eigenvalues
$-\kappa_j^2$ with $0<\kappa_1<\kappa_2<\dots<\kappa_N$, and norming constants $\gamma_j(t)$.
We will write $R(k)=R(k,0)$ and $\gamma_j=\gamma_j(0)$ for the scattering data of the initial condition.
Then the long-time asymptotics can be described by distinguishing the following main regions:

(i).\ 
The soliton region, $x/t>C$ for some $C>0$, in which the solution is asymptotically given by a sum of
one-soliton solutions
\be
q(x,t)\sim -2 \sum_{j=1}^N\frac{\kappa_j^2}{\cosh^2(\kappa_j x -4\kappa_j^3 t -p_j)},
\ee
where the phase shifts are given by
\be
p_j=\frac{1}{2}\log\left(\frac{\gamma_j^2}{2\kappa_j}\prod_{l=j+1}^{N}
\left(\frac{\kappa_l-\kappa_j}{\kappa_l+\kappa_j}\right)^2\right).
\ee
In the case of a pure $N$-soliton solution (i.e., $R(k,t)=0$) this was first established independently by Hirota \cite{hi},
Tanaka \cite{ta}, and Wadati and Toda \cite{wt}.
The general case was first established by \v{S}abat \cite{sh} and by Tanaka \cite{ta2} (see also \cite{es} and \cite{suln}).

(ii).\ 
The self-similar region, $|x/(3t)^{1/3}| \leq C$ for some $C>0$, in which the solution is connected with the Painl\'eve II transcendent.
This was first established by Segur and Ablowitz \cite{as2}.

(iii).\ 
The collisionless shock region, $x<0$ and $C^{-1}<\frac{-x}{(3t)^{1/3}(\log(t))^{2/3}}<C$, for some $C>1$, which
only occurs in the generic case (i.e., when $R(0)=-1$). Here the asymptotics can be given in terms of elliptic functions
as was pointed out by Segur and Ablowitz \cite{as2} with further extensions in Deift, Venakides, and Zhou \cite{dvz}.

(vi).\ 
The similarity region, $x/t<-C$ for some $C>0$, where
\be
q(x,t) \sim \left(\frac{4\nu(k_0) k_0}{3t}\right)^{1/2}\sin(16tk^3_0-\nu(k_0)\log(192tk^3_0)+\delta(k_0)),
\ee 
with
\begin{align*}
\nu(k_0)= &-\frac{1}{2\pi}\log(1-\left\vert R(k_0) \right\vert^2),\\
\delta(k_0) =& \frac{\pi}{4}- \arg(R(k_0))+\arg(\Gamma(\I\nu(k_0))) + 4 \sum_{j=1}^N \arctan\big(\frac{\kappa_j}{k_0}\big)\\ & +
\frac{1}{\pi}\int_{-k_0}^{k_0}\log(\left\vert \zeta-k_0 \right\vert) d\log(1-\left\vert R(\zeta) \right\vert^2).
\end{align*}
Here $k_0=\sqrt{-\frac{x}{12t}}$ denotes the stationary phase point, $R(k)=R(k,t=0)$ the reflection coefficient,
and $\Gamma$ the Gamma function.

Again this was found by Zakharov and Manakov \cite{zama} and (without a precise expression for $\delta(k_0)$
and assuming absence of solitons) by Ablowitz and Segur \cite{as} with further extensions by Buslaev and Sukhanov \cite{bs}
as discussed before.

Our aim here is to use the nonlinear steepest descent method for oscillatory Riemann--Hilbert problems from Deift and Zhou \cite{dz}
and apply it to rigorously establish the  long-time asymptotics in the soliton and similarity regions (Theorem~\ref{thm:asym},
respectively, \ref{thm:asym2}, below). In fact, our main goal is to give a complete and expository introduction to this method.
In addition to providing a streamlined and simplified approach, the following items will be different in comparison with \cite{dz}.
First of all, in the mKdV case considered in \cite{dz} there were no solitons present.
We will add them using the ideas from Deift, Kamvissis, Kriecherbauer, and Zhou \cite{dkkz} following Kr\"uger and Teschl \cite{kt}.
However, in the presence of solitons there is a subtle nonuniqueness issue for the involved Riemann--Hilbert problems
(see e.g.\ \cite[Chap.~38]{bdt}). We will rectify this by imposing an additional symmetry condition and prove that this indeed
restores uniqueness. Secondly, in the mKdV case the reflection coefficient $R(k)$ has always modulus strictly less than one. In the KdV case
this is generically not true and hence terms of the form $R(k)/(1-|R(k)|^2)$ will become singular and cannot be approximated by
analytic functions in the sup norm. We will show how to avoid these terms by using left and right (instead of just right) scattering
data on different parts of the jump contour. Consequently it will be sufficient to approximate the left and right reflection coefficients.
Details can be found in Section~\ref{sec:analapprox}. Moreover,
we obtain precise relations between the error terms and the decay of the initial conditions improving the error estimates obtained
in Schuur \cite{suln} (which are stated in terms of smoothness and decay properties of $R(k)$ and its derivatives).

Overall we closely follow the recent review article \cite{kt2}, where Kr\"uger and Teschl applied these methods to 
compute the long-time asymptotics for the Toda lattice.

For a general result which applies in the case where $R(k)$ has modulus strictly less than one and no solitons are present
we refer to Varzugin \cite{va} and for another recent generalization of the nonlinear steepest descent method to McLaughlin
and Miller \cite{mm}. An alternate approach based on the asymptotic theory of pseudodifferential operators was given
by Budylin and Buslaev \cite{bb}.

Finally, note that if $q(x,t)$ solves the KdV equation, then so does $q(-x,-t)$. Therefore it suffices 
to investigate the case $t\to\infty$.

\section{The Inverse scattering transform and the Riemann--Hilbert problem}
\label{sec:istrhp}

In this section we want to derive the Riemann--Hilbert problem for the KdV equation from
scattering theory. This is essentially classical (compare, e.g., \cite{bdt}) except for two points.
The eigenvalues will be added by appropriate pole conditions which are then
turned into jumps following Deift, Kamvissis, Kriecherbauer, and Zhou \cite{dkkz}. We will
impose an additional symmetry conditions to ensure uniqueness later on following
Kr\"uger and Teschl \cite{kt}.

For the necessary results from scattering theory respectively the inverse
scattering transform for the KdV equation we refer to \cite{mar} (see also \cite{bdt} and \cite{dt}).
We consider real-valued classical solutions $q(x,t)$ of the KdV equation \eqref{kdv}, which
decay rapidly, that is 
\be\label{decay}
\max_{|t|\leq T} \int_{\R} (1+|x|) |q(x,t)| dx < \infty, \qquad \text{for all } T>0.
\ee
Existence of such solutions can for example be established via the inverse
scattering transform if one assumes (cf.\  \cite[Sect.~4.2]{mar}) that the initial condition satisfies
\be
\int_{\R} (1+|x|) \left(|q(x,0)|+|q_x(x,0)|+|q_{xx}(x,0)|+|q_{xxx}(x,0)|\right)dx < \infty.
\ee
Associated with $q(x,t)$ is a self-adjoint Schr{\"o}dinger operator
\begin{equation} \label{defjac}
H(t) = -\frac{d^2}{dx^2}+q(.,t), \qquad \db(H)=H^2(\R) \subset L^2(\R).
\end{equation}
Here $L^2(\R)$ denotes the Hilbert space of square integrable (complex-valued) functions
over $\R$ and $H^k(\R)$ the corresponding Sobolev spaces.

By our assumption \eqref{decay} the spectrum of $H$ consists of an absolutely
continuous part $[0,\infty)$ plus a finite number of eigenvalues  $-\kappa_j^2\in(-\infty,0)$,
$1\le j \le N$. In addition, there exist two Jost solutions $\psi_\pm(k,x,t)$
which solve the differential equation
\be 
H(t) \psi_\pm(k,x,t) = k^2 \psi_\pm(k,x,t), \qquad \im (k)> 0,
\ee
and asymptotically look like the free solutions
\be
\lim_{x \to \pm \infty} \E^{\mp  \I kx} \psi_{\pm}(k,x,t) =1.
\ee
Both $\psi_\pm(k,x,t)$ are analytic for $\im (k) > 0$ and continuous
for $\im (k)\geq 0$.

The asymptotics of the two Jost solutions are
\be\label{eq:psiasym}
\psi_\pm(k,x,t) = \E^{\pm \I kx} \Big(1 + Q_{\pm}(x,t) \frac{1}{2\I k} + O\big(\frac{1}{k^2}\big) \Big),
\ee
as $k \to \infty$ with $\im(k)> 0$, where
\be \label{defQ}
\aligned
Q_+(x,t) &= -\int_x^{\infty} q(y,t) dy , \quad
Q_-(x,t)= -\int_{-\infty}^x q(y,t) dy.
\endaligned
\ee

Furthermore, one has the scattering relations
\be \label{relscat}
T(k) \psi_\mp(k,x,t) =  \ol{\psi_\pm(k,x,t)} +
R_\pm(k,t) \psi_\pm(k,x,t),  \qquad k \in \R,
\ee
where $T(k)$, $R_\pm(k,t)$ are the transmission respectively reflection coefficients.
They have the following well-known properties:

\begin{lemma}
The transmission coefficient $T(k)$ is meromorphic for $\im (k) > 0$ 
with simple poles at $\I \kappa _1 , \dots, \I \kappa_N$ and is continuous up to the real line.
The residues of $T(k)$ are given by
\be\label{eq:resT}
\res_{\I \kappa_j} T(k) = \I  \mu _j(t) \gamma _{+,j}(t)^2 = \I \mu _j \gamma _{+,j}^2,
\ee
where
\be
\gam_{+,j}(t)^{-1} = \lVert \psi _+(\I \kappa_j,.,t)\rVert_2
\ee
and $\psi_+ (\I \kappa_j,x,t) = \mu_j(t) \psi_-(\I \kappa_j,x,t)$.

Moreover,
\be \label{reltrpm} 
T(k) \ol{R_+(k,t)} + \ol{T(k)} R_-(k,t)=0, \qquad |T(k)|^2 + |R_\pm(k,t)|^2=1.
\ee
\end{lemma}

In particular one reflection coefficient, say $R(k,t)=R_+(k,t)$, and one set of
norming constants, say $\gam_j(t)= \gam_{+,j}(t)$, suffices. Moreover,
the time dependence is given by:

\begin{lemma}
The time evolutions of the quantities $R(k,t)$ and $\gam_j(t)$ are given by
\begin{align}
R(k,t) &= R(k) \E^{8 \I k^3 t},\\
\gam_j(t) &= \gam_j \E^{4 \kappa_j^3 t},
\end{align}
where $R(k)=R(k,0)$ and $\gam_j=\gam_j(0)$.
\end{lemma}

We will set up a Riemann--Hilbert problem as follows:
\be\label{defm}
m(k,x,t)= \left\{\begin{array}{c@{\quad}l}
\begin{pmatrix} T(k) \psi_-(k,x,t) \E^{\I kx}  & \psi_+(k,x,t) \E^{-\I kx} \end{pmatrix},
& \im (k) > 0,\\
\begin{pmatrix} \psi_+(-k,x,t) \E^{\I kx} & T(-k) \psi_-(-k,x,t) \E^{-\I kx} \end{pmatrix}, 
& \im(k) < 0.
\end{array}\right.
\ee
We are interested in the jump condition of $m(k,x,t)$ on the real axis $\R$ (oriented
from negative to positive).
To formulate our jump condition we use the following convention:
When representing functions on $\R$, the lower subscript denotes
the non-tangential limit from different sides.
By $m_+(k)$ we denote the limit from above and by $m_-(k)$ the one from below.
Using the notation above implicitly assumes that these limits exist in the sense that
$m(k)$ extends to a continuous function on the real axis.
In general, for an oriented contour $\Sigma$, $m_+(k)$ (resp.\ $m_-(k)$) will denote the limit
of $m(\kappa)$ as $\kappa\to k$ from the positive (resp.\ negative) side of $\Sigma$. Here
the positive (resp.\ negative) side is the one which lies to the left (resp.\ right) as one traverses the contour in the
direction of the orientation.

\begin{theorem}\label{thm:vecrhp}
Let $\mathcal{S}_+(H(0))=\{ R(k),\; k\geq 0; \: (\kappa_j, \gam_j), \: 1\le j \le N \}$ be
the right scattering data of the operator $H(0)$. Then $m(k)=m(k,x,t)$ defined in \eqref{defm}
is a solution of the following vector Riemann--Hilbert problem.

Find a function $m(k)$ which is meromorphic away from the real axis with simple poles at
$\pm\I\kappa_j$ and satisfies:
\begin{enumerate}
\item The jump condition
\be \label{eq:jumpcond}
m_+(k)=m_-(k) v(k), \qquad
v(k)=\begin{pmatrix}
1-|R(k)|^2 & - \ol{R(k)} \E^{-t\Phi(k)} \\ 
R(k) \E^{t\Phi(k)} & 1
\end{pmatrix},
\ee
for $k \in \R$,
\item
the pole conditions
\be\label{eq:polecond}
\aligned
\res_{\I\kappa_j} m(k) &= \lim_{k\to\I\kappa_j} m(k)
\begin{pmatrix} 0 & 0\\ \I \gam_j^2 \E^{t\Phi(\I \kappa_j)}  & 0 \end{pmatrix},\\
\res_{-\I\kappa_j} m(k) &= \lim_{k\to -\I\kappa_j} m(k)
\begin{pmatrix} 0 & - \I \gam_j^2 \E^{t\Phi(\I \kappa_j)} \\ 0 & 0 \end{pmatrix},
\endaligned
\ee
\item
the symmetry condition
\be \label{eq:symcond}
m(-k) = m(k) \sigI,
\ee
\item
and the normalization
\be\label{eq:normcond}
\lim_{\kappa\to\infty} m(\I\kappa) = (1\quad 1).
\ee
\end{enumerate}
Here the phase is given by
\begin{equation}
\Phi(k)= 8 \I k^3+2\I k \frac {x}{t}.
\end{equation}
\end{theorem}

\begin{proof}
The jump condition \eqref{eq:jumpcond} is a simple calculation using the scattering relations
\eqref{relscat} plus \eqref{reltrpm}. The pole conditions follow since $T(k)$ is meromorphic for $\im (k) > 0$
with simple poles at $\I \kappa_j$ and residues given by \eqref{eq:resT}.
The symmetry condition holds by construction and the normalization \eqref{eq:normcond}
is immediate from the following lemma below.
\end{proof}

Observe that the pole condition at $\I \kappa_j$ is sufficient since the one at $-\I \kappa_j$ follows
by symmetry.

Moreover, using
\be\label{eq:asygf}
T(k)\psi_-(k,x,t)\psi_+(k,x,t)=1+\frac{q(x,t)}{2k^2} + O(\frac{1}{k^4})
\ee
as $k \to \infty$ with $\im(k)> 0$ (observe that the right-hand side is just the diagonal Green's functions of $H(t)$ divided by
the free one) we obtain from \eqref{eq:psiasym}

\begin{lemma}\label{lem:asymp}
The function $m(k,x,t)$ defined in \eqref{defm} satisfies
\be\label{eq:asym}
m(k,x,t)=\begin{pmatrix} 1 & 1 \end{pmatrix}+
Q(x,t)\frac{1}{2\I k} \begin{pmatrix} -1 & 1 \end{pmatrix} +O\left(\frac{1}{k^2}\right).
\ee
Here $Q(x,t)=Q_+(x,t)$ is defined in \eqref{defQ}.
\end{lemma}

For our further analysis it will be convenient to rewrite the pole condition as a jump
condition and hence turn our meromorphic Riemann--Hilbert problem into a holomorphic Riemann--Hilbert problem following \cite{dkkz}.
Choose $\eps$ so small that the discs $\left\vert k- \I \kappa_j \right\vert<\eps$ lie inside the upper half plane and
do not intersect. Then redefine $m(k)$ in a neighborhood of $\I \kappa_j$ respectively $- \I \kappa_j$ according to
\be\label{eq:redefm}
m(k) = \begin{cases} m(k) \begin{pmatrix} 1 & 0 \\
-\frac{\I \gamma_j^2 \E^{t\Phi(\I \kappa_j)} }{k- \I \kappa_j} & 1 \end{pmatrix},  &
|k- \I \kappa_j|< \eps,\\
m(k) \begin{pmatrix} 1 & \frac{\I \gamma_j^2 \E^{t\Phi(\I \kappa_j)} }{k+ \I \kappa_j} \\
0 & 1 \end{pmatrix},  &
|k+ \I \kappa_j|< \eps,\\
m(k), & \text{else}.\end{cases}
\ee
 Note that for $\im (k) <0$ we redefined $m(k)$ such that it respects our symmetry \eqref{eq:symcond}. Then a straightforward calculation using
$\res_{\I \kappa} m(k) = \lim_{k\to\I\kappa} (k-\I \kappa)m(k)$ shows:

\begin{lemma}\label{lem:holrhp}
Suppose $m(k)$ is redefined as in \eqref{eq:redefm}. Then $m(k)$ is holomorphic away from
the real axis and the small circles around $\I \kappa_j$ and $-\I\kappa_j$. Furthermore it satisfies \eqref{eq:jumpcond}, \eqref{eq:symcond}, \eqref{eq:normcond}
and the pole condition is replaced by the jump condition
\be \label{eq:jumpcond2}
\aligned
m_+(k) &= m_-(k) \begin{pmatrix} 1 & 0 \\
-\frac{\I \gamma_j^2 \E^{t\Phi(\I\kappa_j)}}{k-\I \kappa_j} & 1 \end{pmatrix},\quad |k-\I \kappa_j|=\eps,\\
m_+(k) &= m_-(k) \begin{pmatrix} 1 & -\frac{\I \gamma_j^2 \E^{t\Phi(\I \kappa_j)}}{k+ \I \kappa_j} \\
0 & 1 \end{pmatrix},\quad |k+ \I \kappa_j|=\eps,
\endaligned
\ee
where the small circle around $\I \kappa_j$ is oriented counterclockwise and the one around $-\I \kappa_j$ is oriented clockwise.
\end{lemma}

Next we turn to uniqueness of the solution of this vector Riemann--Hilbert problem. This will also explain the
reason for our symmetry condition. We begin by observing that if
there is a point $k_1\in\C$, such that $m(k_1)=\rN$, then $n(k)=\frac{1}{k-k_1} m(k)$
satisfies the same jump and pole conditions as $m(k)$. However, it will clearly
violate the symmetry condition! Hence, without the symmetry condition, the solution
of our vector Riemann--Hilbert problem will not be unique in such a situation. Moreover, a look at the
one-soliton solution verifies that this case indeed can happen. 

\begin{lemma}[One-soliton solution]\label{lem:singlesoliton}
Suppose there is only one eigenvalue and that the reflection coefficient vanishes, that is,
$\mathcal{S}_+(H(t))=\{ R(k,t)\equiv 0,\; k\in\R; \: (\kappa, \gam (t)), \kappa>0, \gamma>0 \}$.
Then the unique solution of the Riemann--Hilbert problem \eqref{eq:jumpcond}--\eqref{eq:normcond}
is given by
\begin{align}\label{eq:oss}
m_0(k) &= \begin{pmatrix} f(k) & f(-k) \end{pmatrix} \\
\nn f(k) &= \frac{1}{1+(2\kappa)^{-1} \gamma^2 \E^{t\Phi(\I\kappa)}}
\left(1+\frac{k+\I\kappa}{k-\I\kappa} (2\kappa)^{-1} \gamma^2 \E^{t\Phi(\I\kappa)}\right).
\end{align}
Furthermore, the zero solution is the only solution of the corresponding vanishing problem where
the normalization is replaced by $\lim_{k\to\infty} m(\I k) = (0\quad 0)$.

In particular,
\be
Q(x,t)=\frac{2 \gamma^2 \E^{t\Phi(\I \kappa)}}{1+(2\kappa)^{-1} \gamma^2 \E^{t\Phi(\I\kappa)}}.
\ee
\end{lemma}

\begin{proof}
By assumption the reflection coefficient vanishes and so the jump along the real axis 
disappears. Therefore and by the symmetry condition, we know that the solution is of the form 
$m_0(k) = \begin{pmatrix} f(k) & f(-k) \end{pmatrix}$ where $f(k)$ is meromorphic. Furthermore the function 
$f(k)$ has only a simple pole at $\I\kappa$, so that we can make the ansatz 
$f(k)=C+D\frac{k+\I\kappa}{k-\I\kappa}$. Then the constants $C$ and $D$ are uniquely determined by the 
pole conditions and the normalization.
\end{proof}

In fact, observe $f(k_1)=f(-k_1)=0$ if and only if $k_1=0$ and $2\kappa=\gamma^2 \E^{t\Phi (\I\kappa)} $.
Furthermore, even in the general case $m(k_1)=\rN$ can only occur at $k_1=0$ as the
following lemma shows.

\begin{lemma}\label{lem:resonant}
If $m(k_1) = \rN$ for $m$ defined as in \eqref{defm}, then $k_1  = 0$. Moreover,
the zero of at least one component is simple in this case.
\end{lemma}

\begin{proof}
By \eqref{defm} the condition $m(k_1) = \rN$ implies that the Jost solutions $\psi_-(k,x)$ and
$\psi_+(k,x)$ are linearly dependent or that the transmission coefficient $T(k_1)=0$. This can only happen, at the band edge, $k_1 = 0$
or at an eigenvalue $k_1=\I\kappa_j$.

We begin with the case $k_1=\I\kappa_j$. In this case the derivative of the Wronskian
$W(k)= \psi_+(k,x)\psi_-'(k,x)-\psi_+'(k,x)\psi_-(k,x)$ does not vanish by the well-known formula
$\frac{d}{dk} W(k) |_{k=k_1} = - 2k_1\int_\R \psi_+(k_1,x)  \psi_-(k_1,x) dx \ne 0$. Moreover,
the diagonal Green's function $g(z,x)= W(k)^{-1} \psi_+(k,x) \psi_-(k,x)$ is
Herglotz as a function of $z=-k^2$ and hence can have at most a simple zero at $z=-k_1^2$. 
Since $z\to-k^2$ is conformal away from $z=0$ the same is true as a function of $k$. Hence, if
$\psi_+(\I\kappa_j,x) = \psi_-(\I\kappa_j,x) =0$, both can have at most a simple zero at $k=\I\kappa_j$.
But $T(k)$ has a simple pole at $\I\kappa_j$ and hence $T(k) \psi_-(k,x)$ cannot
vanish at $k=\I\kappa_j$, a contradiction.

It remains to show that one zero is simple in the case $k_1=0$. In fact, 
one can show that $\frac{d}{dk} W(k) |_{k=k_1} \ne 0$ in this case as follows:
First of all note that $\dot{\psi}_\pm(k)$ (where the dot denotes the derivative with respect to
$k$) again solves $H\dot{\psi}_\pm(k_1) = -k_1^2 \dot{\psi}_\pm(k_1)$ if $k_1=0$. Moreover, by
$W(k_1)=0$ we have $\psi_+(k_1) = c\, \psi_-(k_1)$ for some constant $c$ (independent of $x$).
Thus we can compute
\begin{align*}
\dot{W}(k_1) &= W(\dot{\psi}_+(k_1),\psi_-(k_1)) + W(\psi_+(k_1),\dot{\psi}_-(k_1))\\
&= c^{-1} W(\dot{\psi}_+(k_1),\psi_+(k_1)) + c W(\psi_-(k_1),\dot{\psi}_-(k_1))
\end{align*}
by letting $x\to+\infty$ for the first and $x\to-\infty$ for the second Wronskian (in which case we can
replace $\psi_\pm(k)$ by $\E^{\pm \I k x}$),
which gives
\[
\dot{W}(k_1) = -\I(c+c^{-1}).
\]
Hence the Wronskian has a simple zero. But if both functions had more than
simple zeros, so would the Wronskian, a contradiction.
\end{proof}

\section{A uniqueness result for symmetric vector Riemann--Hilbert problems}
\label{sec:uni}

In this section we want to investigate uniqueness for the holomorphic vector Riemann--Hilbert problem 
\begin{align}\nn
& m_+(k) = m_-(k) v(k), \qquad k\in \Sigma,\\ \label{eq:rhp4m}
& m(-k) = m(k) \sigI,\\ \nn
& \lim_{\kappa\to\infty} m(\I\kappa) = \begin{pmatrix} 1 & 1\end{pmatrix},
\end{align}
where we assume

\begin{hypothesis}\label{hyp:sym}
Let $\Sigma$ consist of a finite number of smooth oriented curves in $\C$ such that the distance
between $\Sigma$ and $\{ \I y | y\ge y_0\}$ is positive for some $y_0>0$.
Assume that the contour $\Sigma$ is invariant under $k\mapsto -k$ and $v(k)$ is symmetric
\be
v(-k) = \sigI v(k)^{-1} \sigI,\quad k\in\Sigma.
\ee
Moreover, suppose $\det(v(k))=1$.
\end{hypothesis}

Now we are ready to show that the symmetry condition in fact guarantees uniqueness.

\begin{theorem}
Assume Hypothesis~\ref{hyp:sym}.
Suppose there exists a solution $m(k)$ of the Riemann--Hilbert problem \eqref{eq:rhp4m} for which
$m(k)=\begin{pmatrix} 0 & 0\end{pmatrix}$ can happen at most for $k=0$ in which case
$\limsup_{k\to 0} \frac{k}{m_j(k)}$ is bounded from any direction for $j=1$ or $j=2$.

Then the Riemann--Hilbert problem \eqref{eq:rhp4m} with norming condition replaced by
\be\label{eq:rhp4ma}
\lim_{\kappa\to\infty} m(\I\kappa) = \begin{pmatrix} \alpha & \alpha \end{pmatrix}
\ee
for given $\alpha\in\C$, has a unique solution $m_\alpha(k) = \alpha\, m(k)$.
\end{theorem}

\begin{proof}
Let $m_\alpha(k)$ be a solution of \eqref{eq:rhp4m} normalized according to
\eqref{eq:rhp4ma}. Then we can construct a matrix valued solution via $M=(m, m_\alpha)$ and
there are two possible cases: Either $\det M(k)$ is nonzero for some $k$ or it vanishes
identically.

We start with the first case. Since the determinant of our Riemann--Hilbert problem has no jump
and is bounded at infinity, it is constant. But taking determinants in
\[
M(-k) = M(k) \sigI.
\]
gives a contradiction.

It remains to investigate the case where $\det(M)\equiv 0$. In this case
we have $m_\alpha(k) = \delta(k) m(k)$ with a scalar function $\delta$. Moreover,
$\delta(k)$ must be holomorphic for $k\in\C\backslash\Sigma$ and continuous 
for $k\in\Sigma$ except possibly at the points where $m(k_1) = \rN$. Since it has
no jump across $\Sigma$,
\[
\delta_+(k) m_+(k) = m_{\alpha,+}(k) = m_{\alpha,-}(k) v(k) = \delta_-(k) m_-(k) v(k)
= \delta_-(k) m_+(k),
\]
it is even holomorphic in $\C\backslash\{0\}$ with at most
a simple pole at $k=0$. Hence it must be of the form
\[
\delta(k) = A + \frac{B}{k}.
\]
Since $\delta$ has to be symmetric, $\delta(k) = \delta(-k)$, we obtain $B = 0$. Now, by
the normalization we obtain $\delta(k) = A = \alpha$. This finishes the proof.
\end{proof}

Furthermore, the requirements cannot be relaxed to allow (e.g.) second order
zeros in stead of simple zeros. In fact, if $m(k)$ is a solution for which both components
vanish of second order at, say, $k=0$, then $\ti{m}(k)=\frac{1}{k^2} m(k)$ is a
nontrivial symmetric solution of the vanishing problem (i.e.\ for $\alpha=0$).

By Lemma~\ref{lem:resonant} we have

\begin{corollary}\label{cor:unique}
The solution $m(k)=m(k,x,t)$ found in Theorem~\ref{thm:vecrhp} is the only solution of the
vector Riemann--Hilbert problem \eqref{eq:jumpcond}--\eqref{eq:normcond}.
\end{corollary}

Observe that there is nothing special about $k \to\infty$ where we normalize, any
other point would do as well. However, observe that for the one-soliton solution \eqref{eq:oss},
$f(k)$ vanishes at
\[
k =\I\kappa \frac{1-(2\kappa)^2 \gamma^2 \E^{t\Phi(\I\kappa)}}{1+(2\kappa)^2 \gamma^2 \E^{t\Phi(\I\kappa)}}
\]
and hence the Riemann--Hilbert problem normalized at this point has a nontrivial solution for $\alpha=0$ and
hence, by our uniqueness result, no solution for $\alpha=1$. This shows that
uniqueness and existence are connected, a fact which is not surprising since our
Riemann--Hilbert problem is equivalent to a singular integral equation which is Fredholm of index
zero (see Appendix~\ref{sec:sieq}).

\section{Conjugation and Deformation}
\label{sec:condef}

This section demonstrates how to conjugate our Riemann--Hilbert problem and how to deform our jump 
contour, such that the jumps will be exponentially close to the identity away from the stationary 
phase points. Throughout this and the following section, we will assume that the $R(k)$ has an analytic
extension to a small neighborhood of the real axis. This is for example the case if we assume that our
solution is exponentially decaying. In Section~\ref{sec:analapprox} we will show how to remove this assumption.

For easy reference we note the following result:

\begin{lemma}[Conjugation]\label{lem:conjug}
Assume that $\wti{\Sigma}\subseteq\Sigma$. Let $D$ be a matrix of the form
\be
D(k) = \begin{pmatrix} d(k)^{-1} & 0 \\ 0 & d(k) \end{pmatrix},
\ee
where $d: \C\backslash\wti{\Sigma}\to\C$ is a sectionally analytic function. Set
\be
\ti{m}(k) = m(k) D(k),
\ee
then the jump matrix transforms according to
\be
\ti{v}(k) = D_-(k)^{-1} v(k) D_+(k).
\ee
If $d$ satisfies $d(-k) = d(k)^{-1}$ and $\lim_{\kappa\to\infty} d(\I\kappa)=1$, then the transformation $\ti{m}(k) = m(k) D(k)$
respects our symmetry, that is, $\ti{m}(k)$ satisfies \eqref{eq:symcond} if and only if $m(k)$ does, and our normalization condition.
\end{lemma}

In particular, we obtain
\be
\ti{v} = \begin{pmatrix} v_{11} & v_{12} d^{2} \\ v_{21} d^{-2}  & v_{22} \end{pmatrix},
\qquad k\in\Sigma\backslash\wti{\Sigma},
\ee
respectively
\be
\ti{v} = \begin{pmatrix} \frac{d_-}{d_+} v_{11} & v_{12} d_+ d_- \\
v_{21} d_+^{-1} d_-^{-1}  & \frac{d_+}{d_-} v_{22} \end{pmatrix},
\qquad k\in\wti{\Sigma}.
\ee

In order to remove the poles there are two cases to distinguish. If $\re(\Phi(\I \kappa_j))<0$,
then the corresponding jump is exponentially close to the identity as $t\to\infty$ and there is nothing to do. 
Otherwise we use conjugation to turn the jumps into one with exponentially decaying
off-diagonal entries, again following Deift, Kamvissis, Kriecherbauer, and Zhou \cite{dkkz}.
It turns out that we will have to handle the poles at $\I \kappa_j$ and $-\I\kappa_j$
in one step in order to preserve symmetry and in order to not add additional poles
elsewhere.

\begin{lemma}\label{lem:twopolesinc}
Assume that the Riemann--Hilbert problem for $m$ has jump conditions near $\I\kappa$ and
$-\I\kappa$ given by
\be
\aligned
m_+(k)&=m_-(k)\begin{pmatrix}1&0\\ -\frac{\I\gamma^2}{k-\I\kappa}&1\end{pmatrix}, && 
\left\vert k-\I\kappa \right\vert=\eps, \\
m_+(k)&=m_-(k)\begin{pmatrix}1& -\frac{\I\gamma^2}{k+\I\kappa}\\0&1\end{pmatrix}, && 
\left\vert k+\I\kappa \right\vert=\eps.
\endaligned
\ee
Then this Riemann--Hilbert problem is equivalent to a Riemann--Hilbert problem for $\ti{m}$ which has jump conditions near $\I\kappa$ and
$-\I\kappa$ given by
\begin{align*}
\ti{m}_+(k)&= \ti{m}_-(k)\begin{pmatrix}1& -\frac{(k+\I\kappa)^2}{\I\gamma^2(k-\I\kappa)}\\ 0 &1\end{pmatrix},
&& \left\vert k-\I\kappa \right\vert=\eps, \\
\ti{m}_+(k)&= \ti{m}_-(k)\begin{pmatrix}1& 0 \\ -\frac{(k-\I\kappa)^2}{\I\gamma^2(k+\I\kappa)}&1\end{pmatrix},
&& \left\vert k+\I\kappa \right\vert=\eps,
\end{align*}
and all remaining data conjugated (as in Lemma~\ref{lem:conjug}) by
\be
D(k) = \begin{pmatrix} \frac{k-\I\kappa}{k+\I\kappa} & 0 \\ 0 & \frac{k+\I\kappa}{k-\I\kappa} \end{pmatrix}.
\ee
\end{lemma}

\begin{proof}
To turn $\gam^2$ into $\gam^{-2}$, introduce $D$ by
\[
D(k) = \begin{cases}
\begin{pmatrix} 1 & -\frac{k-\I\kappa}{\I\gamma^2} \\  \frac{\I\gamma^2}{k-\I\kappa} & 0 \end{pmatrix}
\begin{pmatrix} \frac{k-\I\kappa}{k+\I\kappa} & 0 \\ 0 & \frac{k+\I\kappa}{k-\I\kappa} \end{pmatrix}, &  \left\vert k-\I\kappa \right\vert <\eps, \\
\begin{pmatrix} 0 & -\frac{\I\gamma^2}{k+\I\kappa} \\ \frac{k+\I\kappa}{\I\gamma^2} & 1 \end{pmatrix} 
\begin{pmatrix} \frac{k-\I\kappa}{k+\I\kappa} & 0 \\ 0 & \frac{k+\I\kappa}{k-\I\kappa} \end{pmatrix}, & \left\vert k+\I\kappa \right\vert <\eps, \\ 
\begin{pmatrix} \frac{k-\I\kappa}{k+\I\kappa} & 0 \\ 0 & \frac{k+\I\kappa}{k-\I\kappa} \end{pmatrix}, & \text{else},
\end{cases}
\]
and note that $D(k)$ is analytic away from the two circles. Now set $\ti{m}(k) = m(k) D(k)$and note that $D(k)$ is also symmetric. Therefore the jump conditions can be verified by straightforward calculations and Lemma \ref{lem:conjug}.
\end{proof}

The jump along the real axis is of oscillatory type and our aim is to apply 
a contour deformation following \cite{dz} such that all jumps will be moved into regions where the oscillatory terms 
will decay exponentially. Since the jump matrix $v$ contains both $\exp(t\Phi)$ and 
$\exp(-t\Phi)$ we need to separate them in order to be able to move them to different regions 
of the complex plane. 

We recall that the phase of the associated Riemann--Hilbert problem is given by  
\begin{equation} \label{eq:Phi}
\Phi(k)=8\I k^3+2\I k\frac{x}{t}
\end{equation}
and the stationary phase points, $\Phi'(k)=0$, are denoted by $\pm k_0$, where
\be
k_0= \sqrt{-\frac{x}{12t}}.
\ee
For $\frac{x}{t}>0$ we have $k_0\in\I\R$, and for $\frac{x}{t}<0$
we have $k_0\in\R$. For $\frac{x}{t}>0$ we will also need the value
$\kappa_0$ defined via $\re(\Phi(\I\kappa_0))=0$, that is,
\be
 \kappa_0 = \sqrt{\frac{x}{4 t}}>0.
\ee
We will set $\kappa_0=0$ if $\frac{x}{t}<0$ for notational convenience.
A simple analysis shows that for $\frac{x}{t}>0$ we have $0<k_0/\I<\kappa_0$.

As mentioned above we will need the following factorizations of the jump condition \eqref{eq:jumpcond}:
\be
v(k)=b_-(k)^{-1}b_+(k),
\ee
where 
\be
b_-(k)=\begin{pmatrix} 1 & \ol{R(k)}\E^{-t\Phi(k)} \\ 0 & 1 \end{pmatrix}, \qquad
b_+(k)=\begin{pmatrix} 1 & 0 \\ R(k)\E^{t\Phi(k)} & 1\end{pmatrix}.
\ee
for $|k|>\re(k_0)$ and 
\be
v(k)=B_-(k)^{-1}
\begin{pmatrix} 1-\left\vert R(k) \right\vert^2 & 0 \\ 
	0 & \frac{1}{1-\left\vert R(k) \right\vert^2}\end{pmatrix}
B_+(k),
\ee
where
\be
B_-(k)=\begin{pmatrix} 1 & 0\\
	-\frac{R(k)\E^{t\Phi(k)}}{1-\left\vert R(k) \right\vert^2} & 1 \end{pmatrix}, \qquad
B_+(k)=\begin{pmatrix} 1 & -\frac{\ol{R(k)}\E^{-t\Phi(k)}}{1-\left\vert R(k) \right\vert^2}\\
	0 & 1 \end{pmatrix}.
\ee
for $|k|< \re(k_0)$.

To get rid of the diagonal part in the factorization corresponding to 
$\left\vert k \right\vert < \re(k_0)$ and to conjugate the jumps near the eigenvalues we need
the partial transmission coefficient $T(k,k_0)$.

We define the partial transmission coefficient with respect to $k_0$ by
\begin{align}\label{def:Tkk0}
T(k,k_0) &=
\begin{cases}
\prod\limits_{\kappa_j\in(\kappa_0,\infty)}  \frac{k+\I\kappa_j}{k-\I\kappa_j}, & k_0 \in \I\R^+, \\
\prod\limits_{j=1}^{N} \frac{k+\I\kappa_j}{k-\I\kappa_j}
\exp\left(\frac{1}{2\pi\I}\int\limits_{-k_0}^{k_0}\frac{\log(|T(\zeta)|^2)}{\zeta-k}
{d\zeta}\right), & k_0 \in \R^+,
\end{cases}
\end{align}
for $k\in\C\backslash\Sigma(k_0)$, where  $\Sigma(k_0) = [-\re(k_0),\re(k_0)]$ (oriented from left to right). Thus $T(k,k_0)$
is meromorphic for $k\in\C\backslash\Sigma(k_0)$. Note that $T(k,k_0)$ can be computed in terms
of the scattering data since $|T(k)|^2= 1- |R_+(k,t)|^2$. Moreover, we set
\begin{align}
T_1(k_0)
&= \begin{cases}
\sum\limits_{\kappa_j \in (\kappa_0,\infty)} 2\kappa_j, & k_0 \in \I\R^+,\\
\sum\limits_{j=1}^N 2\kappa_j + \frac{1}{2\pi} \int_{-k_0}^{k_0} \log(\left\vert T(\zeta) \right\vert^2)d\zeta , & k_0 \in \R^+. 
\end{cases}
\end{align}
Thus
\begin{align}\label{eq:asyt}
T(k,k_0) = 1+T_1(k_0)\frac{\I}{k}+O\left(\frac{1}{k^2}\right), \quad \text{ as } k\to\infty.
\end{align}

\begin{theorem}\label{thm:part}
The partial transmission coefficient $T(k,k_0)$ is meromorphic in $\C\backslash\Sigma(k_0)$, where
\be
\Sigma(k_0)=[-\re(k_0),\re(k_0)],
\ee
with simple poles at  $\I\kappa_j$ and simple zeros at $-\I\kappa_j$ for all j with $\kappa_0 < \kappa_j$,
and satisfies the jump condition
\be\label{eq:jumpt}
T_+(k,k_0)=T_-(k,k_0)(1-\left\vert R(k) \right\vert^2), \qquad \textnormal{ for } k\in\Sigma(k_0). 
\ee 
Moreover,
\begin{enumerate}
\item
$T(-k,k_0)=T(k,k_0)^{-1}$, $k\in\C\backslash\Sigma(k_0)$, 
\item
$T(-k,k_0)=\ol{T(\ol{k},k_0)}$, $k\in\C$, in particular $T(k,k_0)$ is
real for $k\in\I\R$, and 
\item
if $k_0\in\R^+$ the behaviour near $k=0$ is given by $T(k,k_0) = T(k) (C+ o(1))$ with $C\ne 0$ for $\im(k)\ge 0$.
\end{enumerate}
\end{theorem}

\begin{proof}
That $\I\kappa_j$ are simple poles and $-\I\kappa_j$ are simple zeros is obvious from the 
Blaschke factors and that $T(k,k_0)$ has the given jump follows from Plemelj's formulas. 
(i), (ii), and (iii) are straightforward to check.
\end{proof}

Now we are ready to perform our conjugation step. Introduce
\be
D(k) = \begin{cases}
\begin{pmatrix} 1 & -\frac{k-\I\kappa_j}{\I\gamma_j^2 \E^{t\Phi (\I\kappa_j)}}\\
\frac{\I\gamma_j^2 \E^{t\Phi(\I\kappa_j)}}{k-\I\kappa_j} & 0 \end{pmatrix}
D_0(k), &  |k-\I\kappa_j|<\eps, \: \kappa_0 < \kappa_j,\\
\begin{pmatrix} 0 & -\frac{\I\gamma_j^2 \E^{t\Phi (\I\kappa_j)}}{k+\I\kappa_j} \\
\frac{k+\I\kappa_j}{\I\gamma_j^2 \E^{t\Phi(\I\kappa_j)}} & 1 \end{pmatrix}
D_0(k), & |k+\I\kappa_j|<\eps, \: \kappa_0 < \kappa_j,\\ 
D_0(k), & \text{else},
\end{cases}
\ee
where 
\[
D_0(k) = \begin{pmatrix} T(k,k_0)^{-1} & 0 \\ 0 & T(k,k_0) \end{pmatrix}.
\]
Observe that $D(k)$ respects our symmetry,
\[
D(-k)= \sigI D(k) \sigI.
\]
Now we conjugate our problem using $D(k)$ and set
\be\label{def:mti}
\ti{m}(k)=m(k) D(k).
\ee
Note that even though $D(k)$ might be singular at $k=0$ (if $k_0>0$ and $R(0)=-1$),
$\ti{m}(k)$ is nonsingular since the possible singular behaviour of $T(k,k_0)^{-1}$ from $D_0(k)$
cancels with $T(k)$ in $m(k)$ by virtue of Theorem~\ref{thm:part} (iii).

Then using Lemma~\ref{lem:conjug} and Lemma~\ref{lem:twopolesinc} the jump
corresponding to $\kappa_0 < \kappa_j$ (if any) is given by
\be
\aligned
\ti{v}(k) &= \begin{pmatrix}1& -\frac{k-\I\kappa_j}
{\I\gamma_j^2 \E^{t\Phi (\I\kappa_j)}T(k,k_0)^{-2}}\\ 0 &1\end{pmatrix},
\qquad |k-\I\kappa_j|=\eps, \\
\ti{v}(k) &= \begin{pmatrix}1& 0 \\ -\frac{k+\I\kappa_j}
{\I\gamma_j^2 \E^{t\Phi(\I\kappa_j)} T(k,k_0)^2}&1\end{pmatrix},
\qquad |k+\I\kappa_j|=\eps,
\endaligned
\ee
and corresponding to $\kappa_0>\kappa_j$ (if any) by
\be
\aligned
\ti{v}(k) &= \begin{pmatrix} 1 & 0 \\ -\frac{\I\gamma_j^2 \E^{t\Phi(\I\kappa_j)} T(k,k_0)^{-2}}{k-\I\kappa_j}
 & 1 \end{pmatrix},
\qquad |k-\I\kappa_j|=\eps, \\
\ti{v}(k) &= \begin{pmatrix} 1 & -\frac{\I\gamma_j^2 \E^{t\Phi(\I\kappa_j)} T(k,k_0)^2}{k+\I\kappa_j} \\
0 & 1 \end{pmatrix},
\qquad |k+\I\kappa_j|=\eps.
\endaligned
\ee
In particular, all jumps corresponding to poles, except for possibly one if
$\kappa_j=\kappa_0$, are exponentially close to the identity for $t\to\infty$. In the latter case we will keep the
pole condition for $\kappa_j=\kappa_0$ which now reads
\be
\aligned
\res_{\I\kappa_j} \ti{m}(k) &= \lim_{k\to\I\kappa_j} \ti{m}(k)
\begin{pmatrix} 0 & 0\\ \I\gamma_j^2 \E^{t\Phi(\I\kappa_j)} T(\I\kappa_j,k_0)^{-2}  & 0 \end{pmatrix},\\
\res_{-\I\kappa_j} \ti{m}(k) &= \lim_{k\to -\I\kappa_j} \ti{m}(k)
\begin{pmatrix} 0 & -\I\gamma_j^2 \E^{t\Phi(\I\kappa_j)} T(\I\kappa_j,k_0)^{-2} \\ 0 & 0 \end{pmatrix}.
\endaligned
\ee
Furthermore, the jump along $\R$ is given by
\be
\ti{v}(k) = \begin{cases}
\ti{b}_-(k)^{-1} \ti{b}_+(k), & k\in\R\backslash\Sigma(k_0),\\
\ti{B}_-(k)^{-1} \ti{B}_+(k), & k\in\Sigma(k_0),\\
\end{cases}
\ee
where
\be \label{eq:deftib}
\ti{b}_-(k) = \begin{pmatrix} 1 & \frac{R(-k) \E^{-t\Phi(k)}}{T(-k,k_0)^2} \\ 0 & 1 \end{pmatrix}, \quad
\ti{b}_+(k) = \begin{pmatrix} 1 & 0 \\ \frac{R(k) \E^{t\Phi(k)}}{T(k,k_0)^2} & 1 \end{pmatrix},
\ee
and
\begin{align*}
\ti{B}_-(k) 
=\begin{pmatrix} 1 & 0 \\ -\frac{T_-(k,k_0)^{-2}}{1-\left\vert R(k) \right\vert^2}R(k)\E^{t\Phi(k)} & 1\end{pmatrix}
=\begin{pmatrix} 1 & 0 \\ - \frac{T_-(-k,k_0)}{T_-(k,k_0)} R(k) \E^{t\Phi(k)} & 1 \end{pmatrix}, 
\end{align*}
\begin{align*}
\ti{B}_+(k) 
=\begin{pmatrix} 1 & -\frac{T_+(k,k_0)^2}{1-\left\vert R(k)\right\vert^2}R(-k)\E^{-t\Phi(k)} \\ 0 & 1 \end{pmatrix} 
=\begin{pmatrix} 1 & - \frac{T_+(k,k_0)}{T_+(-k,k_0)} R(-k) \E^{-t\Phi(k)} \\ 0 & 1 \end{pmatrix}.
\end{align*}
Here we have used
\[
R(-k)=\ol{R(k)}, \quad k\in\R, \qquad T_\pm(-k,k_0)=T_\mp(k,k_0)^{-1}, \quad k\in\Sigma(k_0),
\]
and the jump condition \eqref{eq:jumpt} for the partial transmission coefficient $T(k,k_0)$
along $\Sigma(k_0)$ in the last step. This also shows that the matrix entries are bounded for
$k\in\R$ near $k=0$ since $T_\pm(-k,k_0)=\ol{T_\pm(k,k_0)}$.

Since we have assumed that $R(k)$ has an analytic continuation to a neighborhood of the real axis,
we can now deform the jump along $\R$ to move the oscillatory terms into regions where they are
decaying. According to Figure~\ref{fig:signRePhi} there are two cases to distinguish:

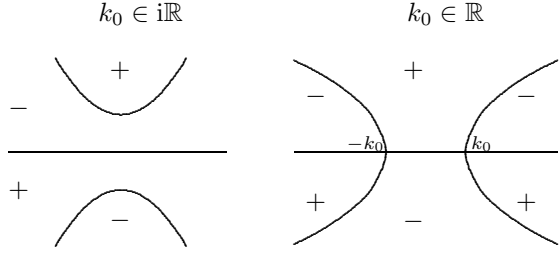
\begin{figure}\centering
\begin{picture}(3.75,3)
\put(0.3,1.25){\line(1,0){2.9}}

\put(1.5,3){$k_0\in\I\R$}

\put(0.3,1.75){$-$}
\put(0.3,0.65){$+$}
\put(1.65,2.25){$+$}
\put(1.65,0.25){$-$}

\curve(0.934,2.5, 1.05,2.3125, 1.3,2., 1.55,1.8125, 1.8,1.75, 2.05,1.8125, 2.3,2., 2.55,2.3125, 2.666,2.5)

\curve(0.934,0., 1.05,0.1875, 1.3,0.5, 1.55,0.6875, 1.8,0.75, 2.05,0.6875, 2.3,0.5, 2.55,0.1875, 2.666,0.)
\end{picture}\quad
\begin{picture}(3.75,3)
\put(0,1.25){\line(1,0){3.5}}

\put(1.5,3){$k_0\in\R$}

\put(0.15,0.5){$+$}
\put(0.15,1.9){$-$}
\put(2.95,0.5){$+$}
\put(2.95,1.9){$-$}
\put(1.45,2.25){$+$}
\put(1.45,0.25){$-$}

\put(0.7,1.3){$\scriptstyle -k_0$}
\put(2.35,1.3){$\scriptstyle k_0$}

\curve(0.,0.025, 0.425,0.25, 0.775,0.5, 1.025,0.75, 1.225,1.25, 1.025,1.75, 0.775,2., 0.425,2.25, 0.,2.47)

\curve(3.5,0.025, 3.075,0.25, 2.725,0.5, 2.475,0.75, 2.275,1.25, 2.475,1.75, 2.725,2., 3.075,2.25, 3.5,2.47)
\end{picture}
\caption{Sign of $\re(\Phi(k))$ for different values of $k_0$}\label{fig:signRePhi}

\end{figure}

\vspace{3mm}\noindent
{\bf Case 1}: $k_0\in\I\R$, $k_0\not=0$:\\[0mm]

\begin{figure}\centering
\begin{picture}(7,5.2)
\put(0.6,2.5){\line(1,0){5.8}}
\put(2,2.5){\vector(1,0){0.4}}
\put(5,2.5){\vector(1,0){0.4}}

\put(6,2.2){$\R$}

\put(0.6,2){\line(1,0){5.8}}
\put(2,2){\vector(1,0){0.4}}
\put(5,2){\vector(1,0){0.4}}

\put(6,1.6){$\Sigma_-$}

\put(0.6,3){\line(1,0){5.8}}
\put(2,3){\vector(1,0){0.4}}
\put(5,3){\vector(1,0){0.4}}

\put(6,3.2){$\Sigma_+$}

\put(0.6,3.5){$\re(\Phi)<0$}
\put(0.6,1.3){$\re(\Phi)>0$}
\put(2.9,4.5){$\re(\Phi)>0$}
\put(2.9,0.5){$\re(\Phi)<0$}

\curvedashes{0.05,0.05}

\curve(1.868,5., 2.1,4.625, 2.6,4., 3.1,3.625, 3.6,3.5, 4.1,3.625, 4.6,4., 5.1,4.625, 5.332,5.)

\curve(1.868,0., 2.1,0.375, 2.6,1., 3.1,1.375, 3.6,1.5, 4.1,1.375, 4.6,1., 5.1,0.375, 5.332,0.)
\end{picture}
\caption{Deformed contour for $k_0\in\I\R^+$}
\label{figure:solreg}
\end{figure}
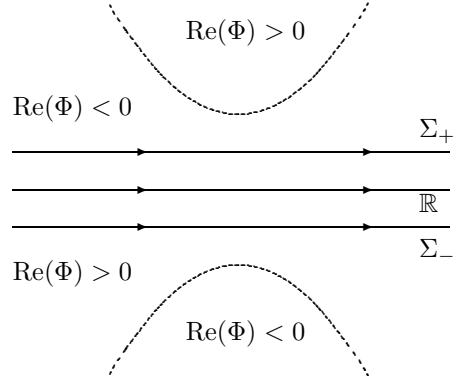

We set $\Sigma_\pm = \{ k\in\C \vert \im(k) = \pm \varepsilon\}$ for some small $\varepsilon$
such that $\Sigma_\pm$ lies in the region with $\pm \re(\Phi(k)) < 0$ and such that the circles around $\pm\I\kappa_j$
lie outside the region in between $\Sigma_-$ and $\Sigma_+$ (see Figure~\ref{figure:solreg}).
Then we can split our jump by redefining $\ti{m}(k)$ according to 
\be
\wha{m}(k) = \begin{cases} \ti{m}(k)\ti{b}_+(k)^{-1} ,  &
0<\im(k)<\varepsilon,\\
\ti{m}(k) \ti{b}_-(k)^{-1},  &
-\varepsilon < \im(k) < 0,\\
\ti{m}(k), & \text{else}.\end{cases}
\ee
Thus the jump along the real axis disappears and the jump along $\Sigma_\pm$ is given by 
\be\label{eq:jumpsolreg}
\wha{v}(k) = \begin{cases} \ti{b}_+(k) , &  k\in\Sigma_+ \\
\ti{b}_-(k)^{-1} , & k\in\Sigma_-.\end{cases}
\ee
All other jumps are unchanged. Note that the resulting Riemann--Hilbert problem still satisfies our symmetry
condition \eqref{eq:symcond}, since we have
\be
\ti{b}_\pm(-k)= \sigI \ti{b}_\mp(k)\sigI.
\ee
By construction the jump along $\Sigma_{\pm}$ is exponentially close to the identity as $t\to\infty$.

\vspace{3mm}\noindent
{\bf Case 2}: $k_0\in\R$, $k_0\not=0$:\\[0mm]

\begin{figure}\centering
\begin{picture}(7,5.2)
\put(0,2.5){\line(1,0){7.0}}
\put(2,2.5){\vector(1,0){0.4}}
\put(5,2.5){\vector(1,0){0.4}}

\put(6.6,2.2){$\R$}

\put(0,2){\line(1,0){2.0}}
\put(2.9,3){\line(1,0){1.2}}
\put(5,2){\line(1,0){2.0}}
\put(1.1,2){\vector(1,0){0.4}}
\put(5.5,2){\vector(1,0){0.4}}
\put(3.3,3){\vector(1,0){0.4}}

\put(1.3,1.5){$\Sigma_-^1$}
\put(5.7,1.5){$\Sigma_-^1$}
\put(3.5,3.3){$\Sigma_+^2$}

\curve(2.,2., 2.2,2.1, 2.4,2.4, 2.45,2.5, 2.5,2.6, 2.7,2.9, 2.9,3.)

\curve(4.1,3., 4.3,2.9, 4.5,2.6, 4.55,2.5, 4.6,2.4, 4.8,2.1, 5.,2.)

\put(0,3){\line(1,0){2.0}}
\put(2.9,2){\line(1,0){1.2}}
\put(5,3){\line(1,0){2.0}}
\put(1.1,3){\vector(1,0){0.4}}
\put(5.5,3){\vector(1,0){0.4}}
\put(3.3,2){\vector(1,0){0.4}}

\curve(2.,3., 2.2,2.9, 2.4,2.6, 2.45,2.5, 2.5,2.4, 2.7,2.1, 2.9,2.)

\curve(4.1,2., 4.3,2.1, 4.5,2.4, 4.55,2.5, 4.6,2.6, 4.8,2.9, 5.,3.)

\put(1.3,3.3){$\Sigma_+^1$}
\put(5.7,3,3){$\Sigma_+^1$}
\put(3.5,1.5){$\Sigma_-^2$}

\put(0.3,1.0){$\scriptstyle\re(\Phi)>0$}
\put(0.3,3.8){$\scriptstyle\re(\Phi)<0$}
\put(5.9,1.0){$\scriptstyle\re(\Phi)>0$}
\put(5.9,4.0){$\scriptstyle\re(\Phi)<0$}
\put(2.9,4.5){$\scriptstyle\re(\Phi)>0$}
\put(2.9,0.5){$\scriptstyle\re(\Phi)<0$}

\put(2.6,2.6){$\scriptstyle -k_0$}
\put(4.7,2.6){$\scriptstyle k_0$}

\curvedashes{0.05,0.05}

\curve(0.,0.05, 0.85,0.5, 1.55,1., 2.05,1.5, 2.45,2.5, 2.05,3.5, 1.55,4., 0.85,4.5, 0.,4.94)

\curve(7.,0.05, 6.15,0.5, 5.45,1., 4.95,1.5, 4.55,2.5, 4.95,3.5, 5.45,4., 6.15,4.5, 7.,4.94)
\end{picture}
\caption{Deformed contour for $k_0\in\R^+$}
\label{figure:simreg}
\end{figure}
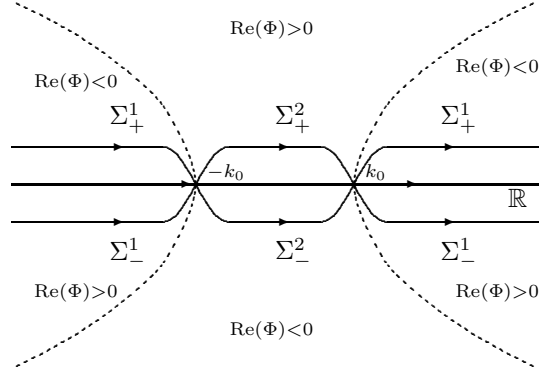

We set $\Sigma_\pm=\Sigma_\pm^1 \cup\Sigma_\pm^2$ according to Figure~\ref{figure:simreg} chosen such
that the circles around $\pm\I\kappa_j$ lie outside the region in between $\Sigma_-$ and $\Sigma_+$. 
Again note that $\Sigma_\pm^1$ respectively $\Sigma_\pm^2$ lie in the region with $\pm\re(\Phi(k))<0$.
Then we can split our jump by redefining $\ti{m}(k)$ according to
\be
\wha{m}(k)= \begin{cases} \ti{m}(k)\ti{b}_+(k)^{-1} , &  k\textnormal{ between } \R \textnormal{ and } \Sigma_+^1, \\
\ti{m}(k)\ti{b}_-(k)^{-1} , & k \textnormal{ between } \R \textnormal{ and } \Sigma_-^1,\\
\ti{m}(k)\ti{B}_+(k)^{-1} , & k \textnormal{ between } \R \textnormal{ and } \Sigma_+^2,\\
\ti{m}(k)\ti{B}_-(k)^{-1} , & k \textnormal{ between } \R \textnormal{ and } \Sigma_-^2,\\
\ti{m}(k) , & \textnormal{else}.
 \end{cases}
\ee

One checks that the jump along $\R$ disappears and the jump along $\Sigma_{\pm}$ is given by 
\be
\wha{v}(k)=\begin{cases} \ti{b}_+(k) , & k\in\Sigma_+^1,\\
\ti{b}_-(k)^{-1} , & k\in\Sigma_-^1,\\
\ti{B}_+(k) , & k\in\Sigma_+^2,\\
\ti{B}_-(k)^{-1} , &  k\in\Sigma_-^2. \end{cases}
\ee

All other jumps are unchanged. Again the resulting Riemann--Hilbert problem still satisfies our symmetry 
condition \eqref{eq:symcond} and the jump along $\Sigma_{\pm}\backslash\{k_0,-k_0\}$ is exponentially decreasing as $t\to\infty$

\begin{theorem}\label{thm:asym}
Assume
\be\label{decayl}
\int_{\R} (1+|x|)^{1+l} |q(x,0)| dx < \infty
\ee
for some integer $l\ge 1$ and abbreviate by $c_j= 4 \kappa_j^2$
the velocity of the $j$'th soliton determined by $\re(\Phi(\I \kappa_j))=0$.
Then the asymptotics in the soliton region, $x/t \geq C $ for some
$C>0$, are as follows:

Let $\eps > 0$ sufficiently small such that the intervals
$[c_j-\eps,c_j+\eps]$, $1\le j \le N$, are disjoint and lie inside $\R^+$. 

If $|\frac{x}{t} - c_j|<\eps$ for some $j$, one has
\begin{align}
\int_x^{\infty} q(y,t) dy &= -4 \sum_{i=j+1}^N \kappa_i -
\frac{2\gamma_j^2(x,t)}{1+ (2\kappa_j)^{-1} \gamma_j^2(x,t)} + O(t^{-l}),
\end{align}
respectively
\begin{align}
q(x,t)& = \frac{-4\kappa_j\gamma_j^2(x,t)}{(1+(2\kappa_j)^{-1}\gamma_j^2(x,t))^2} +O(t^{-l}),
\end{align}
where
\be
\gam_j^2(x,t) = \gamma_j^2 \E^{-2\kappa_j x + 8 \kappa_j^3 t} \prod_{i=j+1}^N \left(\frac{\kappa_i-\kappa_j}{\kappa_i+\kappa_j}\right)^2.
\ee

If $|\frac{x}{t} -c_j| \geq \eps$, for all $j$, one has
\begin{align}
\int_x^{\infty} q(y,t) dy &= -4 \sum_{\kappa_i \in (\kappa_0,\infty)} \kappa_i + O(t^{-l}), \qquad \kappa_0=\sqrt{\frac{x}{4t}},
\end{align}
respectively
\begin{align}
q(x,t)& = O(t^{-l}).
\end{align}
\end{theorem}

\begin{proof}
Since $\wha{m}(k) = \ti{m}(k)$ for $k$ sufficiently far away from $\R$ equations \eqref{eq:asym}, \eqref{def:mti}, and \eqref{eq:asyt} imply
the following asymptotics
\be\label{eq:asyhm}
\wha{m}(k) = \begin{pmatrix} 1 & 1 \end{pmatrix} +
\left(-2T_1(k_0) + Q(x,t)\right) \frac{1}{2\I k} \begin{pmatrix}-1 & 1 \end{pmatrix} + O\left(\frac{1}{k^2}\right).
\ee
By construction, the jump along $\Sigma_\pm$ is exponentially decreasing as $t\to\infty$. 
Hence we can apply Theorem~\ref{thm:remcontour} as follows:

If $|\frac{x}{t}-c_j|>\varepsilon$ (resp. $|\kappa_0^2-\kappa_j^2|>\varepsilon$) for all $j$ we can choose $\gamma^t=0$ and $w^t=\wha{w}$
in Theorem~\ref{thm:remcontour}. Since $\wha{w}$ is exponentially small as $t\to\infty$, the
solutions of the associated Riemann--Hilbert problems only differ by $O(t^{-l})$ for any $l\ge 1$. 
Comparing $m_0=  \begin{pmatrix} 1 & 1 \end{pmatrix}$ with the above asymptotics shows $Q_+(x,t)= 2T_1(k_0) +O(t^{-l})$.

If $|\frac{x}{t} -c_j| < \varepsilon$ (resp. $|\kappa_0^2-\kappa_j^2|<\varepsilon$) for some $j$, we choose 
$\gamma^t=\gamma_k(x,t)$ and $w^t=\wha{w}$ in Theorem~\ref{thm:remcontour}, where 
\[
\gamma_j^2(x,t)=\gamma_j^2 \E^{t\Phi(\I\kappa_j)}T(\I\kappa_j, \I\frac{\kappa_j}{\sqrt{3}})^{-2}
= \gamma_j^2 \E^{-2\kappa_j x + 8 \kappa_j^3 t} \prod_{i=j+1}^N \left(\frac{\kappa_i-\kappa_j}{\kappa_i+\kappa_j}\right)^2.
\]
As before we conclude that $\wha{w}$ is exponentially small
and so the associated solutions of the Riemann--Hilbert problems only differ by $O(t^{-l})$.
From Lemma \ref{lem:singlesoliton}, we have the one-soliton solution 
$m_0(k) = \begin{pmatrix} f(k) & f(-k) \end{pmatrix}$
with
\[
f(k)=  \frac{1}{1+(2\kappa_j)^{-1}\gamma_j^2(x,t)}\big(1+\frac{k+\I\kappa_j}{k-\I\kappa_j}(2\kappa_j)^{-1}
\gamma_j^2(x,t)\big).
\]
As before, comparing with the above asymptotics shows
\[
Q(x,t)= 2 T_1(k_0) + \frac{2\gamma_j^2(x,t)}{1+(2\kappa_j)^{-1}\gamma_j^2(x,t)}+O(t^{-l}).
\]
To see the second part just use \eqref{eq:asygf} in place of \eqref{eq:asym}.
This finishes the proof in the case where $R(k)$ has an analytic extensions. We will remove this assumption
in Section~\ref{sec:analapprox} thereby completing the proof.
\end{proof}

Since the one-soliton solution is exponentially decaying away from its minimum, we also obtain the form stated in the introduction:

\begin{corollary}
Assume \eqref{decayl}, then the asymptotic in the soliton region, 
$x/t \geq C$ for some $C>0$, is given by
\be
q(x,t)= -2\sum_{j=1}^{N}\frac{\kappa_j^2}{\cosh^2(\kappa_j x-4\kappa_j^3 t -p_j)} + O(t^{-l}),
\ee
where
\be
p_j=\frac{1}{2}\log\left(\frac{\gamma_j^2}{2\kappa_j}\prod_{i=j+1}^{N}
\left(\frac{\kappa_i-\kappa_j}{\kappa_i+\kappa_j}\right)^2\right).
\ee
\end{corollary}

\section{Reduction to a Riemann--Hilbert problem on a small cross}
\label{sec:simrhp}

In the previous section we have seen that for $k_0\in\R$ we can reduce everything to a 
Riemann--Hilbert problem for $\wha{m}(k)$ such that the jumps are exponentially close to the identity
except in small neighborhoods of the stationary phase points $k_0$ and $-k_0$. Hence we 
need to continue our investigation of this case in this section. 

Denote by $\Sigma^c(\pm k_0)$ the parts of $\Sigma_+\cup \Sigma_-$ inside a small neighborhood of $\pm k_0$. 
We will now show that solving the two problems on the small crosses $\Sigma^c(k_0)$ respectively $\Sigma^c(-k_0)$
will lead us to the solution of our original problem.

In fact, the solution of both crosses can be reduced to the following model problem:
Introduce the cross $\Sigma = \Sigma_1 \cup\dots\cup \Sigma_4$ (see Figure~\ref{fig:contourcross}) by
\begin{align}
\nn \Sigma_1 & = \{u \E^{-\I\pi/4},\,u\in [0,\infty)\} &
\Sigma_2 & = \{u \E^{\I\pi/4},     \,u\in [0,\infty)\} \\
\Sigma_3 & = \{u \E^{3\I\pi/4},    \,u\in [0,\infty)\} &
\Sigma_4 & = \{u \E^{-3\I\pi/4},   \,u\in [0,\infty)\}.
\end{align}
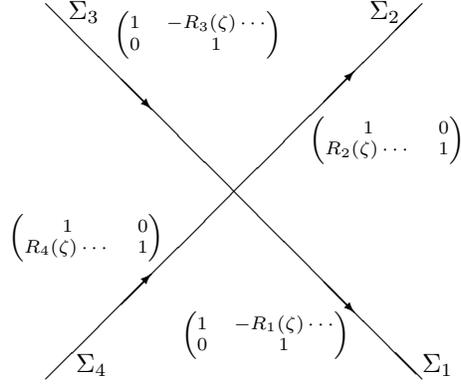
\begin{figure}
\begin{picture}(7,5.2)

\put(1,5){\line(1,-1){5}}
\put(2,4){\vector(1,-1){0.4}}
\put(4.7,1.3){\vector(1,-1){0.4}}
\put(1,0){\line(1,1){5}}
\put(2,1){\vector(1,1){0.4}}
\put(4.7,3.7){\vector(1,1){0.4}}

\put(6.0,0.1){$\Sigma_1$}
\put(5.3,4.8){$\Sigma_2$}
\put(1.3,4.8){$\Sigma_3$}
\put(1.4,0.1){$\Sigma_4$}

\put(2.8,0.5){$\scriptsize\begin{pmatrix} 1 & - R_1(\zeta) \cdots\\ 0 & 1 \end{pmatrix}$}
\put(4.5,3.1){$\scriptsize\begin{pmatrix} 1 & 0 \\ R_2(\zeta) \cdots & 1 \end{pmatrix}$}
\put(1.9,4.5){$\scriptsize\begin{pmatrix} 1 & - R_3(\zeta) \cdots \\ 0 & 1 \end{pmatrix}$}
\put(0.5,1.8){$\scriptsize\begin{pmatrix} 1 & 0 \\ R_4(\zeta) \cdots  & 1 \end{pmatrix}$}

\end{picture}
\caption{Contours of a cross}
\label{fig:contourcross}
\end{figure}%
Orient $\Sigma$ such that the real part of $k$ increases
in the positive direction. Denote by $\mathbb{D} = \{\zeta,\,|\zeta|<1\}$ the
open unit disc. Throughout this section $\zeta^{\I\nu}$ will denote
the function $\E^{\I \nu \log(\zeta)}$, where the branch cut of the logarithm is chosen along
the negative real axis $(-\infty,0)$.

Introduce the following jump matrices ($v_j$ for $\zeta\in\Sigma_j$)
\begin{align}
\nn v_1 &= \begin{pmatrix} 1 & - R_1(\zeta) \zeta^{2\I\nu} \E^{- t \Phi(\zeta)} \\ 0 & 1 \end{pmatrix}, &
v_2 &= \begin{pmatrix} 1 & 0 \\ R_2(\zeta) \zeta^{-2\I\nu} \E^{t \Phi(\zeta)} & 1 \end{pmatrix},  \\
v_3 &= \begin{pmatrix} 1 & - R_3(\zeta) \zeta^{2\I\nu} \E^{- t \Phi(\zeta)} \\ 0 & 1 \end{pmatrix}, &
v_4 &= \begin{pmatrix} 1 & 0 \\ R_4(\zeta) \zeta^{-2\I\nu} \E^{t \Phi(\zeta)}  & 1 \end{pmatrix}
\end{align}
and consider the RHP given by
\begin{align}\label{eq:rhpcross}
M_+(\zeta) &= M_-(\zeta) v_j(\zeta), && \zeta\in\Sigma_j,\quad j=1,2,3,4,\\ \nn
M(\zeta) &\to \id, && \zeta\to \infty.
\end{align}
The solution is given in the following theorem of Deift and Zhou \cite{dz} (for a proof of
the version stated below see Kr\"uger and Teschl \cite{kt2}).

\begin{theorem}[\cite{dz}]\label{thm:solcross}
Assume there is some $\rho_0>0$ such that $v_j(\zeta)=\id$ for $|\zeta|>\rho_0$. Moreover,
suppose that within $|\zeta|\le\rho_0$ the following estimates hold:
\begin{enumerate}
\item
The phase satisfies $\Phi(0)=\I\Phi_0\in\I\R$, $\Phi'(0) = 0$, $\Phi''(0) = \I$ and
\begin{align}\label{estPhi}
\pm \re\big(\Phi(\zeta)\big) &\geq \frac{1}{4} |\zeta|^2,\quad
\begin{cases} + & \mbox{for } \zeta\in\Sigma_1\cup\Sigma_3,\\ - &\mbox{else},\end{cases}\\ \label{estPhi2}
|\Phi(\zeta) - \Phi(0) - \frac{\I \zeta^2}{2}| &\leq C |\zeta|^3.
\end{align}
\item
There is some $r\in\mathbb{D}$ and constants $(\alpha, L) \in (0,1] \times (0,\infty)$
such that $R_j$, $j=1,\dots,4$, satisfy H\"older conditions of the form
\begin{align}\nn
\abs{R_1(\zeta) - \ol{r}} &\leq L |\zeta|^\alpha, &
\abs{R_2(\zeta) - r} &\leq L |\zeta|^\alpha, \\\label{holdcondrj}
\abs{R_3(\zeta) - \frac{\ol{r}}{1-\abs{r}^2}} &\leq L |\zeta|^\alpha, &
\abs{R_4(\zeta) - \frac{r}{1-\abs{r}^2}} &\leq L |\zeta|^\alpha.
\end{align}
\end{enumerate}
Then the solution of the RHP \eqref{eq:rhpcross} satisfies
\be
M(\zeta) = \id + \frac{1}{\zeta} \frac{\I}{t^{1/2}} \begin{pmatrix} 0 & -\beta \\ \ol{\beta} & 0 \end{pmatrix}
+ O(t^{- \frac{1 + \alpha}{2}}),
\ee
for $|\zeta|>\rho_0$, where
\be
\beta = \sqrt{\nu} \E^{\I(\pi/4-\arg(r)+\arg(\Gamma(\I\nu)))} \E^{-\I t \Phi_0} t^{-\I\nu},
\qquad \nu = - \frac{1}{2\pi} \log(1 - |r|^2).
\ee
Furthermore, if $R_j(\zeta)$ and $\Phi(\zeta)$ depend on some parameter, the error term is uniform
with respect to this parameter as long as $r$ remains within a compact subset of $\mathbb{D}$
and the constants in the above estimates can be chosen independent of the parameters.
\end{theorem}

\begin{theorem}[Decoupling]\label{thm:decoup}
Consider the Riemann--Hilbert problem 
\begin{align}\nn
& m_+(k)=m_-(k)v(k), \qquad k\in\Sigma, \\
& \lim_{\kappa\to\infty} m(\I\kappa) = \begin{pmatrix} 1 & 1\end{pmatrix},
\end{align}
with $\det(v)\ne 0$ and let $0<\alpha<\beta \leq 2\alpha, \rho(t)\to\infty$ be given.

Suppose that for every sufficiently small $\varepsilon >0$ both the $L^2$ and the $L^{\infty}$
norms of $v$ are $O(t^{-\beta})$ away from some $\varepsilon$ neighborhoods of some points 
$k_j\in\Sigma$, $1\leq j\leq n$. Moreover, suppose that the solution of the matrix problem with jump 
$v(k)$ restricted to the $\varepsilon$ neighborhood of $k_j$ has a solution which satisfies 
\be
M_j(k)=\id +\frac{1}{\rho(t)^{\alpha}}\frac{M_j}{k-k_j}+O(\rho(t)^{-\beta}), \qquad 
\left\vert k-k_j \right\vert >\varepsilon.
\ee
Then the solution $m(k)$ is given by 
\be
m(k)=\begin{pmatrix} 1 & 1 \end{pmatrix} + 
\frac{1}{\rho(t)^{\alpha}}\begin{pmatrix} 1 & 1 \end{pmatrix} \sum_{j=1}^n \frac{M_j}{k-k_j} +
O(\rho(t)^{-\beta}),
\ee
where the error term depends on the distance of $k$ to $\Sigma$.
\end{theorem}

\begin{proof}
In this proof we will use the theory developed in Appendix \ref{sec:sieq} with $m_0(k)=\begin{pmatrix} 1 & 1 \end{pmatrix}$ 
and the usual Cauchy kernel $\Omega_\infty(s,k)=\id \frac{ds}{s-k}$. Moreover, since
symmetry is not important, we will consider $C_w$ on $L^2(\Sigma)$ rather than restricting
it to the symmetric subspace $L^2_s(\Sigma)$. Here $w_\pm= \pm(b_\pm-\id)$ correspond to some factorization $v=b_-^{-1}b_+$
of $v$ (e.g., $b_-=\id$ and $b_+=v$). Assume that $m(k)$ exists, then the same arguments
as in the appendix show that 
\[
m(k) =\begin{pmatrix} 1 & 1 \end{pmatrix} + \frac{1}{2\pi\I} \int_\Sigma \mu(s) (w_+(s) + w_-(s)) \Omega_\infty(s,k),
\]
where $\mu$ solves
\[
(\id - C_w) (\mu-\begin{pmatrix} 1 & 1 \end{pmatrix}) = C_w \begin{pmatrix} 1 & 1 \end{pmatrix}.
\]
Introduce $\hat{m}(k)$ by 
\be
\aligned
\hat{m}(k)=\begin{cases} m(k)M_j(k)^{-1}, & \left\vert k-k_j \right\vert \leq 2\varepsilon,\\
m(k), & \text{else}. \end{cases}
\endaligned
\ee
The Riemann--Hilbert problem for $\hat{m}(k)$ has jumps given by
\be
\hat{v}(k)=\begin{cases} M_j(k)^{-1}, & \left\vert k-k_j \right\vert = 2\varepsilon,\\
M_j(k)v(k)M_j(k)^{-1}, & k\in\Sigma, 
\varepsilon < \left\vert k-k_j \right\vert < 2\varepsilon,\\
\id, & k\in\Sigma, \left\vert k-k_j \right\vert < \varepsilon,\\
v(k), & \text{else}. \end{cases}
\ee
By assumption the jumps are $\id+O(\rho(t)^{-\alpha})$ on the circles $|k-k_j| = 2\varepsilon$ and even
$\id +O(\rho(t)^{-\beta})$ on the rest (both in the $L^2$ and $L^\infty$ norms). 
In particular, we infer that $(\id-C_{\hat{w}})^{-1}$ exists for sufficiently large $t$ and using the 
Neumann series to estimate $(\hat{\mu}-\begin{pmatrix} 1 & 1 \end{pmatrix})=(\id-C_{\hat{w}})^{-1}C_{\hat{w}}
\begin{pmatrix} 1 & 1 \end{pmatrix}$ (cf.\ the proof of Theorem~\ref{thm:remcontour}) we obtain
\be
\left\Vert \hat{\mu}-\begin{pmatrix} 1 & 1 \end{pmatrix} \right\Vert_2 \leq  \frac{c\Vert \hat{w} \Vert_2}{1-c\Vert \hat{w} \Vert_{\infty}}
= O (\rho(t)^{-\alpha}).
\ee
Thus we conclude
\be
\aligned
m(k) & = \begin{pmatrix} 1 & 1 \end{pmatrix} + 
\frac{1}{2\pi\I}\int_{\hat{\Sigma}} \hat{\mu}(s)\hat{w}(s) \frac{ds}{s-k}\\
& = \begin{pmatrix} 1 & 1\end{pmatrix} +
\frac{1}{2\pi\I}\sum_{j=1}^n \int_{|s-k_j| = 2\varepsilon}\hat{\mu}(s)(M_j(s)^{-1}-\id)\frac{ds}{s-k}+O(\rho(t)^{-\beta})\\
& = \begin{pmatrix} 1 & 1 \end{pmatrix}- \rho(t)^{-\alpha}\begin{pmatrix} 1 & 1 \end{pmatrix} \frac{1}{2\pi\I}
\sum_{j=1}^n M_j \int_{|s-k_j| = 2\varepsilon}\frac{1}{s-k_j}\frac{ds}{s-k}+O (\rho(t)^{-\beta})\\
& = \begin{pmatrix} 1 & 1 \end{pmatrix} + \rho(t)^{-\alpha} \begin{pmatrix} 1 & 1 \end{pmatrix}
\sum_{j=1}^n \frac{M_j}{k-k_j} + O(\rho(t)^{-\beta}),
\endaligned
\ee
and hence the claim is proven.
\end{proof}

Now let us turn to the solution of the problem on 
$\Sigma^c(k_0)=(\Sigma_+\cup\Sigma_-)\cap\{k \vert \left\vert k-k_0 \right\vert<\varepsilon\}$
for some small $\varepsilon > 0$. Without loss we can also deform our contour slightly such
that $\Sigma^c(k_0)$ consists of two straight lines.
Next, note
\[
\Phi(k_0)=-16 \I k_0^3, \qquad \Phi''(k_0)=48\I k_0.
\]

As a first step we make a change of coordinates
\be\label{eq:zeta}
\zeta=\sqrt{48 k_0}(k-k_0) , \qquad k=k_0+\frac{\zeta}{\sqrt{48 k_0}}
\ee
such that the phase reads $\Phi(k)=\Phi(k_0)+\frac{\I}{2}\zeta^2+O(\zeta^3)$. 

Next we need the behavior of our jump matrix near $k_0$, that is, the behavior of $T(k,k_0)$ near $k_0$.

\begin{lemma}
Let $k_0\in\R$, then
\be
T(k,k_0)=\left(\frac{k-k_0}{k+k_0}\right)^{\I\nu} \ti{T}(k,k_0),
\ee
where $\nu=-\frac{1}{\pi}\log(\left\vert T(k_0) \right\vert)>0$ and the branch cut
of the logarithm is chosen along the negative real axis.
Here
\be\label{tiT}
\ti{T}(k,k_0)=\prod_{j=1}^{N} \frac{k+\I\kappa_j}{k-\I\kappa_j} 
\exp\left(\frac{1}{2\pi\I}\int\limits_{-k_0}^{k_0} 
\log\left( \frac{\left\vert T(\zeta) \right\vert ^2}{\left\vert T(k_0) \right\vert ^2} \right)
\frac{1}{\zeta -k}{d\zeta}\right)
\ee
is H\"older continuous of any exponent less than $1$ at the stationary phase point $k=k_0$ and satisfies
$|\ti{T}(k_0,k_0)|=1$.
\end{lemma}

\begin{proof}
First of all observe that
\be
\exp\left(\frac{1}{2\pi \I}\int_{-k_0}^{k_0} \log(\left\vert T(k_0) \right\vert^2)
\frac{1}{\zeta-k}d\zeta \right) = \left(\frac{k-k_0}{k+k_0}\right)^{\I \nu}. 
\ee
H\"older continuity of any exponent less than $1$ is well-known (cf.\ \cite{mu}).
\end{proof}

If $k(\zeta)$ is defined as in \eqref{eq:zeta} and $0<\alpha<1$, then there is an $L>0$ such that
\be
\left\vert T(k(\zeta ),k_0)-\zeta^{\I\nu}\ti{T}(k_0,k_0)\E^{-\I\nu \log(2 k_0\sqrt{48 k_0})}
\right\vert\leq L\left\vert \zeta \right\vert^{\alpha},
\ee
where the branch cut of $\zeta^{\I\nu}$ is chosen along the negative real axis.

We also have 
\be
\left\vert R(k(\zeta))-R(k_0) \right\vert \leq L\left\vert \zeta \right\vert^{\alpha}
\ee
and thus the assumptions of Theorem~\ref{thm:solcross} are satisfied with
\be
r=R(k_0)\ti{T}(k_0,k_0)^{-2}\E^{2\I \nu \log(2k_0 \sqrt{48k_0})}
\ee
and $\nu=-\frac{1}{2\pi}\log(1-|R(z_0)|^2)$ since $|r|=|R(z_0)|$.
Therefore we can conclude that the solution on $\Sigma^c(k_0)$ is given by 
\be
\aligned
M_1^c(k) & =\id + \frac{1}{\zeta}\frac{\I}{t^{1/2}}\begin{pmatrix} 0 &  -\beta \\ \ol{\beta} & 0 \end{pmatrix} + O(t^{-\alpha})\\
& = \id+ \frac{1}{\sqrt{48k_0}(k-k_0)}\frac{\I}{t^{1/2}} \begin{pmatrix} 0 & -\beta \\
\ol{\beta} & 0 \end{pmatrix} + O(t^{-\alpha}),
\endaligned
\ee
where $\beta$ is given by
\be
\aligned
\beta & = \sqrt{\nu}\E^{\I (\pi/4 -\arg(r) +\arg(\Gamma (\I\nu)))}\E^{-t\Phi(k_0)}t^{-\I\nu}\\
& = \sqrt{\nu}\E^{\I(\pi/4- \arg(R(k_0)) +\arg(\Gamma(\I\nu)))}\ti{T}(k_0,k_0)^2 (192k_0^3)^{-\I\nu}\E^{-t\Phi(k_0)}t^{-\I\nu}
\endaligned
\ee
and $1/2<\alpha<1$.

We also need the solution $M_2^c(k)$ on $\Sigma^c(-k_0)$. We make the following ansatz, which
is inspired by the symmetry condition for the vector Riemann--Hilbert problem, outside the two small crosses:
\begin{align}\nn
M_2^c(k) &= \sigI M_1^c(-k) \sigI\\
 &= \id - \frac{1}{\sqrt{48k_0}(k+k_0)}\frac{\I}{t^{1/2}}\begin{pmatrix} 0 & \ol{\beta} \\ 
-\beta & 0 \end{pmatrix} +O(t^{-\alpha}).
\end{align}
Applying Theorem~\ref{thm:decoup} yields the following result:

\begin{theorem}\label{thm:asym2}
Assume 
\be
\int_{\R} (1+|x|)^6 |q(x,0)| dx < \infty,
\ee
then the asymptotics in the similarity region, $x/t \leq -C$ for some $C>0$, are given by 
\begin{align}\nn
\int_x^{\infty} q(y,t)dy= & -4 \sum_{\kappa_j\in (\kappa_0,\infty)} \kappa_j - \frac{1}{\pi} \int_{-k_0}^{k_0} \log(\left\vert T(\zeta) \right\vert^2)d\zeta\\\label{eq:simasymp2}
& -\sqrt{\frac{\nu(k_0)}{3 k_0 t}}
\cos (16tk_0^3 -\nu(k_0) \log(192 t k_0^3)+\delta(k_0) )+O(t^{-\alpha})
\end{align}
respectively
\be\label{eq:simasymp}
\aligned
q(x,t)=\sqrt{\frac{4\nu(k_0) k_0}{3t}}\sin(16tk_0^3-\nu(k_0)\log(192 t k_0^3)+\delta(k_0))+O(t^{-\alpha})
\endaligned
\ee
for any $1/2<\alpha <1$.
Here $k_0= \sqrt{-\frac{x}{12t}}$ and
\begin{align}
\nu(k_0)= & -\frac{1}{\pi} \log(\left\vert T(k_0) \right\vert),\\ \nn
\delta(k_0)= & \frac{\pi}{4}- \arg(R(k_0))+\arg(\Gamma(\I\nu(k_0)))+4 \sum_{j=1}^N \arctan\big(\frac{\kappa_j}{k_0}\big)\\
&  -\frac{1}{\pi}\int_{-k_0}^{k_0}\log\left(\frac{\left\vert T(\zeta) \right\vert^2}{\left\vert T(k_0) \right\vert^2}\right)\frac{1}{\zeta-k_0}d\zeta.
\end{align}
\end{theorem}

\begin{proof}
By Theorem~\ref{thm:decoup} we have
\begin{align*}
\wha{m}(k)
= & \begin{pmatrix} 1 & 1 \end{pmatrix}+\frac{1}{\sqrt{48k_0}}\frac{\I}{t^{1/2}}
\left(\frac{1}{k-k_0}\begin{pmatrix} \ol{\beta} & -\beta \end{pmatrix}-\frac{1}{k+k_0}
\begin{pmatrix} -\beta & \ol{\beta} \end{pmatrix}\right) +O(t^{-\alpha})\\
 = & \begin{pmatrix} 1 & 1 \end{pmatrix} +\frac{1}{\sqrt{48k_0}}\frac{\I}{t^{1/2}}\frac{1}{k}
\left( \sum_{l=0}^{\infty}\left(\frac{k_0}{k}\right)^l \begin{pmatrix} \ol{\beta} &
-\beta \end{pmatrix}-\sum_{l=0}^{\infty}\left(-\frac{k_0}{k}\right)^l \begin{pmatrix} -\beta & \ol{\beta} \end{pmatrix}\right) \\
& +O(t^{-\alpha}),
\end{align*}
which leads to 
\[
Q(x,t) = 2T_1(k_0)+\frac{4}{\sqrt{48k_0}}\frac{1}{t^{1/2}}(\re(\beta))+ O(t^{-\alpha})
\]
upon comparison with \eqref{eq:asyhm}.
Using the fact that $\left\vert \beta/\sqrt{\nu} \right\vert =1$ proves the first claim.

To see the second part, as in the proof of Theorem~\ref{thm:asym}, just use \eqref{eq:asygf} in place of \eqref{eq:asym},
which shows
\[
q(x,t)=\sqrt{\frac{4k_0}{3t}}\im(\beta) + O(t^{-\alpha}).
\]
This finishes the proof in the case where $R(k)$ has an analytic extensions. We will remove this assumption
in Section~\ref{sec:analapprox} thereby completing the proof.
\end{proof}

Equivalence of the formula for $\delta(k_0)$ given in the previous theorem with the one given in the
introduction follows after a simple integration by parts.

\begin{remark}
Formally the equation \eqref{eq:simasymp} for $q$ can be obtained by differentiating the equation 
\eqref{eq:simasymp2} for $Q$ with respect to $x$. That this step is admissible could be shown as in 
Deift and Zhou \cite{dz2}, however our approach avoids this step.
\end{remark}

\begin{remark}
Note that Theorem~\ref{thm:decoup} does not require uniform boundedness of the associated 
integral operator in contradistinction to Theorem~\ref{thm:remcontour}. We only need the 
knowledge of the solution in some small neighborhoods. However it cannot be used in the 
soliton region, because our solution is not of the form $\id+o(1)$.
\end{remark}

\section{Analytic Approximation}
\label{sec:analapprox}

In this section we want to present the necessary changes in the case where the
reflection coefficient does not have an analytic extension. The idea is to
use an analytic approximation and to split the reflection in an analytic part plus
a small rest. The analytic part will be moved to the complex plane while the rest
remains on the real axis. This needs to be done in such a way that the rest
is of $O(t^{-1})$ and the growth of the analytic part can be controlled by the
decay of the phase.

In the soliton region a straightforward splitting based on the Fourier transform
\be
R(k) =\int_{\R}\E^{\I kx}\hat{R}(x)dx
\ee
will be sufficient. It is well-known that our decay assumption \eqref{decayl} implies
$\hat{R} \in L^1(\R)$ and the estimate (cf.\ \cite[Sect.~3.2]{mar})
\be
|\hat{R}(-2x)| \le const \int_x^\infty q(r) dr, \qquad x \ge 0,
\ee
implies $x^l \hat{R}(-x) \in L^1(0,\infty)$.

\begin{lemma}\label{lem:analapprox}
Suppose $\hat{R} \in L^1(\R)$, $x^l \hat{R}(-x) \in L^1(0,\infty)$ and let $\varepsilon, \beta>0$ be given.
Then we can split the reflection coefficient according to
$R(k)= R_{a,t}(k) + R_{r,t}(k)$ such that $R_{a,t}(k)$ is analytic in $0 <\im(k)< \varepsilon$  and 
\be
|R_{a,t}(k) \E^{-\beta t} | = O(t^{-l}), \quad 0< \im(k) < \varepsilon, \qquad
|R_{r,t}(k)| = O(t^{-l}), \quad k\in\R.
\ee
\end{lemma}

\begin{proof}
We choose $R_{a,t}(k) = \int_{-K(t)}^\infty \E^{\I kx}\hat{R}(x)dx $ with $K(t) =
\frac{\beta_0}{\varepsilon} t$ for some positive $\beta_0<\beta$. Then, for 
$0< \im(k) <\varepsilon$,
\[
\left\vert R_{a,t}(k)\E^{-\beta t} \right\vert 
\leq \E^{-\beta t} \int_{-K(t)}^\infty \vert \hat{R}(x) \vert \E^{-\im(k) x}dx
\leq \E^{-\beta t}\E^{K(t)\varepsilon}\Vert \hat{R} \Vert_1
= \Vert \hat{R} \Vert_1 \E^{-(\beta-\beta_0)t},
\]
which proves the first claim.
Similarly, for $\im(k)=0$,
\[
\vert R_{r,t}(k) \vert 
=\int_{K(t)}^{\infty} \frac{x^l \vert \hat{R}(-x)\vert}{x^l} dx
\leq \frac{\|x^l \hat{R}(-x)\|_{L^1(0,\infty)}}{K(t)^l} 
\leq \frac{const }{t^{l}}  
\]
\end{proof}

To apply this lemma in the soliton region $k_0\in\I\R^+$ we choose
\be
\beta= \min_{\im(k)= \varepsilon} -\re(\Phi(k))>0.
\ee
and split $R(k) =  R_{a,t}(k) + R_{r,t}(k)$ according to Lemma~\ref{lem:analapprox} to obtain
\be
\ti{b}_\pm(k) = \ti{b}_{a,t,\pm}(k) \ti{b}_{r,t,\pm}(k) = \ti{b}_{r,t,\pm}(k) \ti{b}_{a,t,\pm}(k).
\ee
Here $\ti{b}_{a,t,\pm}(k)$, $\ti{b}_{r,t,\pm}(k)$ denote the matrices obtained from $\ti{b}_\pm(k)$
as defined in \eqref{eq:deftib} by replacing $R(k)$ with $R_{a,t}(k)$, $R_{r,t}(k)$, respectively.
Now we can move the analytic parts into the complex plane as in Section~\ref{sec:condef}
while leaving the rest on the real axis. Hence, rather then \eqref{eq:jumpsolreg}, the jump now reads
\be
\hat{v}(k) = \begin{cases}
\ti{b}_{a,t,+}(k), & k\in\Sigma_+, \\
\ti{b}_{a,t,-}(k)^{-1}, & k\in\Sigma_-,\\
\ti{b}_{r,t,-}(k)^{-1} \ti{b}_{r,t,+}(k), & k\in\R.
\end{cases}
\ee
By construction we have $\hat{v}(k)= \id + O(t^{-l})$ on the whole contour and the rest follows as
in Section~\ref{sec:condef}.

In the similarity region not only $\ti{b}_\pm$ occur as jump matrices but also $\ti{B}_\pm$.
These matrices $\ti{B}_\pm$ have at first sight more complicated off diagonal entries, but a closer look shows
that they have indeed the same form.  To remedy
this we will rewrite them in terms of left rather then right scattering data. For this purpose
let us use the notation $R_r(k) \equiv R_+(k)$ for the right and $R_l(k) \equiv R_-(k)$ for the
left reflection coefficient. Moreover, let $T_r(k,k_0) \equiv T(k,k_0)$ be the right and
$T_l(k,k_0) \equiv T(k)/T(k,k_0)$ be the left partial transmission coefficient.

With this notation we have
\be
\ti{v}(k) = \begin{cases}
\ti{b}_-(k)^{-1} \ti{b}_+(k), \qquad \re(k_0)< \left\vert k \right\vert,\\
\ti{B}_-(k)^{-1} \ti{B}_+(k), \qquad \re(k_0) > \left\vert k \right\vert,\\
\end{cases}
\ee
where
\[
\ti{b}_-(k) = \begin{pmatrix} 1 & \frac{R_r(-k) \E^{-t\Phi(k)}}{T_r(-k,k_0)^2} \\ 0 & 1 \end{pmatrix}, \quad
\ti{b}_+(k) = \begin{pmatrix} 1 & 0 \\ \frac{R_r(k) \E^{t\Phi(k)}}{T_r(k,k_0)^2}& 1 \end{pmatrix},
\]
and
\begin{align*}
\ti{B}_-(k) &= \begin{pmatrix} 1 & 0 \\ - \frac{T_{r,-}(k,k_0)^{-2}}{|T(k)|^2} R_r(k) \E^{t\Phi(k)} & 1 \end{pmatrix}, \\
\ti{B}_+(k) &= \begin{pmatrix} 1 & - \frac{T_{r,+}(k,k_0)^2}{|T(k)|^2} R_r(-k) \E^{-t\Phi(k)} \\ 0 & 1 \end{pmatrix}.
\end{align*}
Using \eqref{reltrpm} we can further write
\be
\ti{B}_-(k)= \begin{pmatrix} 1 & 0 \\ \frac{R_l(-k) \E^{t\Phi(k)}}{T_l(-k,k_0)^2} & 1 \end{pmatrix}, \quad
\ti{B}_+(k)= \begin{pmatrix} 1 & \frac{R_l(k) \E^{-t\Phi(k)}}{T_l(k,k_0)^2} \\ 0 & 1 \end{pmatrix}.
\ee
Now we can proceed as before with $\ti{B}_\pm(k)$ as with $\ti{b}_\pm(k)$ by splitting $R_l(k)$ rather than $R_r(k)$.

In the similarity region we need to take the small vicinities of the stationary phase points into account. Since
the phase is cubic near these points, we cannot use it to dominate the exponential growth of the analytic
part away from the unit circle. Hence we will take the phase as a new variable and use the Fourier transform
with respect to this new variable. Since this change of coordinates is singular near the stationary phase points,
there is a price we have to pay, namely, requiring additional smoothness for $R(k)$. In this respect note that
\eqref{decayl} implies $R(k)\in C^l(\R)$ (cf.\ \cite{kl}). We begin with

\begin{lemma}
Suppose $R(k)\in C^5(\R)$. Then we can split $R(k)$ according to
\be
R(k) = R_0(k) + (k-k_0)(k+k_0) H(k), \qquad k \in \Sigma(k_0),
\ee
where $R_0(k)$ is a real rational function in $k$ such that $H(k)$ vanishes  
at $k_0$, $-k_0$ of order three and
has a Fourier transform
\be
H(k)=\int_{\R}\hat{H}(x) \E^{x\Phi(k)}dx,
\ee
with $x\hat{H}(x)$ integrable.
\end{lemma}

\begin{proof}
We can construct a rational function, which satisfies $f_n(-k)=\ol{f_n(k)}$ for $k\in\R$,
by making the ansatz 
$f_n(k)=\frac{k_0^{2n+4}+1}{k^{2n+4}+1}\sum_{j=0}^n\frac{1}{j!(2 k_0)^j}(\alpha_j+\I\beta_j \frac{k}{k_0})(k-k_0)^j(k+k_0)^j$.
Furthermore we can choose $\alpha_j$, $\beta_j\in\R$ for $j=1,\dots,n$, such that we can match the values of $R$
and its first four derivatives at $k_0$, $-k_0$ at $f_n(k)$. Thus we will set $R_0(k)=f_4(k)$ with $\alpha_0=\re(R(k_0))$,
$\beta_0=\im(R(k_0))$, and so on. Since $R_0(k)$ is integrable we infer that $H(k)\in C^4(\R)$ and it vanishes
together with its first three derivatives at $k_0$, $-k_0$.

Note that $\Phi(k)/\I=8 (k^3-3 k_0^2 k)$ is a polynomial of order  
three which has a maximum at $-k_0$
and a minimum at $k_0$. Thus the phase $\Phi(k)/\I$ restricted to $ 
\Sigma(k_0)$ gives a one to one coordinate transform
$\Sigma(k_0) \to [\Phi(k_0)/\I, \Phi(-k_0)/\I]=[-16k_0^3,16k_0^3]$   
and we can hence express $H(k)$ in this new coordinate
(setting it equal to zero outside this interval). The coordinate  
transform locally looks like a cube root near $k_0$ and $-k_0$,
however, due to our assumption that $H$ vanishes there, $H$ is still  
$C^2$ in this new coordinate and the Fourier transform
with respect to this new coordinates exists and has the required  
properties.
\end{proof}

Moreover, as in Lemma~\ref{lem:analapprox} we obtain:

\begin{lemma}
Let $H(k)$ be as in the previous lemma. Then we can split $H(k)$ according to
$H(k)= H_{a,t}(k) + H_{r,t}(k)$ such that $H_{a,t}(k)$ is analytic in the region $\re(\Phi(k))<0$
and 
\be
|H_{a,t}(k) \E^{\Phi(k) t/2} | = O(1), \: \re(\Phi(k))<0, \im(k) \le 0, \quad
|H_{r,t}(k)| = O(t^{-1}), \: k\in\R.
\ee
\end{lemma}

\begin{proof}
We choose $H_{a,t}(k) = \int_{-K(t)}^{\infty}\hat{H}(x)\E^{x \Phi(k)}dx$ with $K(t) = t/2$.
Then we can conclude as in Lemma~\ref{lem:analapprox}:
\begin{align*}
\vert H_{a,t}(k) \E^{\Phi(k) t/2} \vert
\leq \Vert \hat{H}(x) \Vert_1 \vert \E^{-K(t) \Phi(k)+\Phi(k) t/2} \vert
\leq\Vert \hat{H}(x) \Vert_1
\leq const 
\end{align*} 
and 
\begin{align*}
\vert H_{r,t}(k)\vert
& \leq \int_{-\infty}^{-K(t)} \vert \hat{H}(x) \vert dx
\leq const \sqrt{\int_{-\infty}^{-K(t)}  \frac{1}{x^4}  dx}
\leq const  \frac{1}{K(t)^{3/2}} \leq const \frac{1}{t}.
\end{align*}
\end{proof}

By construction $R_{a,t}(k) = R_0(k) + (k-k_0)(k+k_0) H_{a,t}(k)$ will satisfy the required
Lipschitz estimate in a vicinity of the stationary phase points (uniformly in $t$) and all
jumps will be $\id+O(t^{-1})$. Hence we can proceed as in Section~\ref{sec:simrhp}.

\appendix

\section{Singular integral equations}
\label{sec:sieq}

In this section we show how to transform a meromorphic vector Riemann--Hilbert problem
with simple poles at $\I\kappa$, $-\I\kappa$,
\begin{align}\nn
& m_+(k) = m_-(k) v(k), \qquad k\in \Sigma,\\ \label{eq:rhp5m}
& \res_{\I\kappa} m(k) = \lim_{k\to\I\kappa} m(k)
\begin{pmatrix} 0 & 0\\ \I\gamma^2  & 0 \end{pmatrix}, \quad
 \res_{-\I\kappa} m(k) = \lim_{k\to -\I\kappa} m(k)
\begin{pmatrix} 0 & -\I\gamma^2  \\ 0 & 0 \end{pmatrix},\\ \nn
& m(-k) = m(k) \sigI,\\ \nn
& \lim_{k\to\infty} m(\I\ k) = \begin{pmatrix} 1 & 1\end{pmatrix}
\end{align}
into a singular integral equation.
Since we require the symmetry condition for our Riemann--Hilbert
problem we need to adapt the usual Cauchy kernel to preserve this symmetry.
Moreover, we keep the single soliton as an inhomogeneous term which will play
the role of the leading asymptotics in our applications.

The classical Cauchy-transform
of a function $f:\Sigma\to \C$ which is square integrable is the
analytic function $C f: \C\backslash\Sigma\to\C$ given by
\be
C f(k) = \frac{1}{2\pi\I} \int_{\Sigma} \frac{f(s)}{s - k} ds,\qquad k\in\C\backslash\Sigma.
\ee
Denote the tangential boundary values from both sides (taken possibly
in the $L^2$-sense --- see e.g.\ \cite[eq.\ (7.2)]{deiftbook}) by $C_+ f$ respectively $C_- f$.
Then it is well-known that $C_+$ and $C_-$ are bounded operators $L^2(\Sigma)\to L^2(\Sigma)$, which satisfy $C_+ - C_- = \id$ (see e.g. \cite{deiftbook}). Moreover, one has
the Plemelj--Sokhotsky formula (\cite{mu})
\be
C_\pm = \frac{1}{2} (\I H \pm \id),
\ee
where
\be
H f(k) = \frac{1}{\pi} \dashint_\Sigma \frac{f(s)}{k-s} ds,\qquad k\in\Sigma,
\ee
is the Hilbert transform and $\dashint$ denotes the principal value integral.

In order to respect the symmetry condition we will restrict our attention to
the set $L^2_{s}(\Sigma)$ of square integrable functions $f:\Sigma\to\C^{2}$ such that
\be\label{eq:sym}
f(-k) = f(k) \sigI.
\ee
Clearly this will only be possible if we require our jump data to be symmetric as well:

\begin{hypothesis}\label{hyp:sym2}
Suppose the jump data $(\Sigma,v)$ satisfy the following assumptions:
\begin{enumerate}
\item
$\Sigma$ consist of a finite number of smooth oriented finite curves in $\C$
which intersect at most finitely many times with all intersections being transversal.
\item
The distance between $\Sigma$ and $\{ \I y | y\ge y_0\}$ is positive for some $y_0>0$ and
$\pm\I\kappa\not\in\Sigma$.
\item
$\Sigma$ is invariant under $k\mapsto -k$ and is oriented such that under the
mapping $k\mapsto -k$ sequences converging from the positive sided to $\Sigma$
are mapped to sequences converging to the negative side.
\item
The jump matrix $v$ is invertible and can be factorized according to $v = b_-^{-1} b_+ =
(\id-w_-)^{-1}(\id+w_+)$, where $w_\pm = \pm(b_\pm-\id)$ satisfy
\be\label{eq:wpmsym}
w_\pm(-k) = -\sigI w_\mp(k) \sigI,\quad k\in\Sigma.
\ee
\item
The jump matrix satisfies
\begin{align}\nn
\|w\|_\infty &= \|w_+\|_{L^\infty(\Sigma)} + \|w_-\|_{L^\infty(\Sigma)}<\infty,\\ \label{hyp:normw}
\| w \|_2 &= \| w_+ \|_{L^2(\Sigma)} + \| w_- \|_{L^2(\Sigma)}<\infty.
\end{align}
\end{enumerate}
\end{hypothesis}

Next we introduce the Cauchy operator
\be
(C f)(k) = \frac{1}{2\pi\I} \int_\Sigma f(s) \Omega_\kappa(s,k)
\ee
acting on vector-valued functions $f:\Sigma\to\C^{2}$.
Here the Cauchy kernel is given by
\be
\Omega_\kappa(s,k) =
\begin{pmatrix} \frac{k+\I\kappa}{s+\I\kappa} \frac{1}{s-k} & 0 \\
0 & \frac{k-\I\kappa}{s-\I\kappa} \frac{1}{s-k} \end{pmatrix} ds =
\begin{pmatrix} \frac{1}{s-k} - \frac{1}{s+\I\kappa} & 0 \\
0 & \frac{1}{s-k} - \frac{1}{s-\I\kappa} \end{pmatrix} ds,
\ee
for some fixed $\I\kappa\notin\Sigma$. In the case $\kappa=\infty$ we set
\be
\Omega_\infty(s,k) =
\begin{pmatrix} \frac{1}{s-k} & 0 \\
0 & \frac{1}{s-k} \end{pmatrix} ds.
\ee
and one easily checks the symmetry property:
\be\label{eq:symC}
\Omega_\kappa(-s,-k) = \sigI \Omega_\kappa(s,k) \sigI.
\ee
The properties of $C$ are summarized in the next lemma.

\begin{lemma}
Assume Hypothesis~\ref{hyp:sym2}.
The Cauchy operator $C$ has the properties, that the boundary values
$C_\pm$ are bounded operators $L^2_s(\Sigma) \to L^2_s(\Sigma)$
which satisfy
\be\label{eq:cpcm}
C_+ - C_- = \id
\ee
and
\be\label{eq:Cnorm}
(Cf)(-\I\kappa) = (0\quad\ast), \qquad (Cf)(\I\kappa) = (\ast\quad 0).
\ee
Furthermore, $C$ restricts to $L^2_{s}(\Sigma)$, that is
\be
(C f) (-k) = (Cf)(k) \sigI,\quad k\in\C\backslash\Sigma
\ee
for $f\in L^2_{s}(\Sigma) \text{ or } L^{\infty}_{s}(\Sigma)$ and by \eqref{eq:wpmsym} we also have
\be \label{eq:symcpm}
C_\pm(f w_\mp)(-k) = C_\mp(f w_\pm)(k) \sigI,\quad k\in\Sigma.
\ee
\end{lemma}

\begin{proof}
Everything follows from \eqref{eq:symC} and the fact that $C$ inherits all properties from
the classical Cauchy operator.
\end{proof}

We have thus obtained a Cauchy transform with the required properties.
Following Section 7 and 8 of \cite{bc}, we can solve our Riemann--Hilbert problem using this
Cauchy operator.

Introduce the operator $C_w: L_s^2(\Sigma)\to L_s^2(\Sigma)$ by
\be
C_w f = C_+(fw_-) + C_-(fw_+),\quad f\in L^2_s(\Sigma).
\ee
By our hypothesis \eqref{hyp:normw} $C_w$ is also well-defined as operator from
$L_s^\infty(\Sigma)\to L_s^2(\Sigma)$ and we have
\be\label{estnormcw}
\|C_w\|_{L^2_s\to L^2_s} \le const \|w\|_\infty \quad\text{respectively}\quad
\|C_w\|_{L^\infty_s\to L^2_s} \le const \|w\|_2.
\ee
Furthermore recall from Lemma~\ref{lem:singlesoliton} that the unique 
solution corresponding to $v\equiv \id$ is given by
\begin{align}\nn
& m_0(k)= \begin{pmatrix} f(k) & f(-k) \end{pmatrix}, \\
& f(k) = \frac{1}{1+(2\kappa)^{-1}\gamma^2\E^{t\Phi(\I\kappa)}}
\left( 1+\frac{k+\I\kappa}{k-\I\kappa} (2\kappa)^{-1} \gamma^2 \E^{t\Phi(\I\kappa)}\right).
\end{align}
Observe that for $\gam=0$ we have $f(k)=1$ and for $\gam=\infty$ we have
$f(k)= \frac{k+\I\kappa}{k-\I\kappa}$. In particular, $f(k)$ is uniformly bounded for all $\gam\in[0,\infty]$
if $|k-\I\kappa|>\eps$.

Then we have the next result.

\begin{theorem}\label{thm:cauchyop}
Assume Hypothesis~\ref{hyp:sym2}.

Suppose $m$ solves the Riemann--Hilbert problem \eqref{eq:rhp5m}. Then
\be\label{eq:mOm}
m(k) = (1-c_0)m_0(k) + \frac{1}{2\pi\I} \int_\Sigma \mu(s) (w_+(s) + w_-(s)) \Omega_\kappa(s,k),
\ee
where
\[
\mu = m_+ b_+^{-1} = m_- b_-^{-1} \quad\mbox{and}\quad
c_0= \left( \frac{1}{2\pi\I} \int_\Sigma \mu(s) (w_+(s) + w_-(s)) \Omega_\kappa(s,\infty) \right)_{\!1}.
\]
Here $(m)_j$ denotes the $j$'th component of a vector.
Furthermore, $\mu$ solves
\be\label{eq:sing4muc}
(\id - C_w) (\mu(k)-(1-c_0)m_0(k) ) = (1-c_0) C_w m_0(k)
\ee

Conversely, suppose $\ti{\mu}$ solves 
\be\label{eq:sing4mu}
(\id - C_w) (\ti{\mu}(k)-m_0(k)) =C_w m_0(k),
\ee
and
\[
\ti{c}_0= \left( \frac{1}{2\pi\I} \int_\Sigma \ti{\mu}(s) (w_+(s) + w_-(s)) \Omega_\kappa(s,\infty) \right)_{\!1} \ne -1,
\]
then $m$ defined via \eqref{eq:mOm}, with $(1-c_0)=(1+\ti{c}_0)^{-1}$ and $\mu=(1+\ti{c}_0)^{-1}\ti{\mu}$,
solves the Riemann--Hilbert problem \eqref{eq:rhp5m} and $\mu= m_\pm b_\pm^{-1}$.
\end{theorem}

\begin{proof}
If $m$ solves \eqref{eq:rhp5m} and we set $\mu = m_\pm b_\pm^{-1}$,
then $m$ satisfies an additive jump given by
\[
m_+ - m_- = \mu (w_+ + w_-).
\]
Hence, if we denote the left hand side of \eqref{eq:mOm} by $\ti{m}$, both functions satisfy the same additive
jump. Furthermore, Hypothesis~\ref{hyp:sym} implies that $\mu$ is symmetric and hence so is $\ti{m}$. 
Using \eqref{eq:Cnorm} we also see that $\ti{m}$ satisfies the same pole conditions as $m_0$. In summary,
$m-\ti{m}$ has no jump and solves \eqref{eq:rhp5m} with $v\equiv \id$ except for the normalization which is given by 
$\lim_{k\to\infty} m(\I k)-\ti{m}(\I k)=(0 \quad 0)$. Hence Lemma~\ref{lem:singlesoliton} implies $m-\ti{m} =0$.

Moreover, if $m$ is given by \eqref{eq:mOm}, then \eqref{eq:cpcm} implies
\begin{align} \label{eq:singtorhp}
m_\pm &= (1-c_0)m_0 + C_\pm(\mu w_-) + C_\pm(\mu w_+) \\
\nn &= (1-c_0)m_0 + C_w(\mu) \pm \mu w_\pm \\
\nn &= (1-c_0)m_0 - (\id - C_w) \mu + \mu b_\pm.
\end{align}
From this we conclude that $\mu = m_\pm b_\pm^{-1}$ solves \eqref{eq:sing4muc}.

Conversely, if $\ti{\mu}$ solves \eqref{eq:sing4mu}, then set
\[
\ti{m}(k) = m_0(k) + \frac{1}{2\pi\I} \int_\Sigma \ti{\mu}(s) (w_+(s) + w_-(s)) \Omega_\zeta(s,k),
\]
and the same calculation as in \eqref{eq:singtorhp} implies
$\ti{m}_\pm= \ti{\mu} b_\pm$, which shows that $m = (1+\ti{c}_0)^{-1} \ti{m}$ solves the
Riemann--Hilbert problem \eqref{eq:rhp5m}.
\end{proof}
 
\begin{remark}
In our case $m_0(k)\in L^\infty_s(\Sigma)$, but $m_0(k)$ is not square integrable and so 
$\mu\in L^2_s(\Sigma)+L^\infty_s(\Sigma)$ in general.
\end{remark} 
 
Note also that in the special case $\gamma=0$ we have $m_0(k)= \rI$ and
we can choose $\kappa$ as we please, say $\kappa=\infty$ such that $c_0=\ti{c}_0=0$
in the above theorem.

Hence we have a formula for the solution of our Riemann--Hilbert problem $m(k)$ in terms of
$m_0+(\id - C_w)^{-1} C_w m_0$ and this clearly raises the question of bounded
invertibility of $\id - C_w$ as  a map from $L_s^2(\Sigma)\to L_s^2(\Sigma)$.
 This follows from Fredholm theory (cf.\ e.g. \cite{zh}):

\begin{lemma}
Assume Hypothesis~\ref{hyp:sym2}.
The operator $\id-C_w$ is Fredholm of index zero,
\be
\ind(\id-C_w) =0.
\ee
\end{lemma}

By the Fredholm alternative, it follows that to show the bounded invertibility of $\id-C_w$
we only need to show that $\ker (\id-C_w) =0$.

We are interested in comparing a Riemann--Hilbert problem for which $\|w\|_\infty$ and $\|w\|_2$
is small with the one-soliton problem.
For such a situation we have the following result:

\begin{theorem}\label{thm:remcontour}
Fix a contour $\Sigma$ and choose $\kappa$, $\gam=\gam^t$, $v^t$ depending on some parameter $t\in\R$ such that
Hypothesis~\ref{hyp:sym2} holds.

Assume that $w^t$ satisfies
\be
\|w^t \|_{\infty} \leq \rho(t) \text{ and } 
\|w^t \|_2 \leq \rho(t)
\ee
for some function $\rho(t) \to 0$ as $t\to\infty$. Then $(\id-C_{w^t})^{-1}: L^2_s(\Sigma)\to L^2_s(\Sigma)$ exists
for sufficiently large $t$ and the solution $m(k)$ of the Riemann--Hilbert problem \eqref{eq:rhp5m} differs
from the one-soliton solution $m_0^t(k)$ only by $O(\rho(t))$, where the error term depends on the distance of $k$ to
$\Sigma \cup \{\pm\I\kappa\}$.
\end{theorem}

\begin{proof}
By \eqref{estnormcw} we conclude that
\[
\|C_{w^t}\|_{L^2_s \to L^2_s} =O(\rho(t)) \quad\text{respectively}\quad
\|C_{w^t}\|_{L^\infty_s \to L^2_s} =O(\rho(t))
\]
Thus, by the Neumann series, we infer that $(\id-C_{w^t})^{-1}$ exists for sufficiently large $t$ and
\[
\|(\id-C_{w^t})^{-1}-\id\|_{L^2_s \to L^2_s} = O(\rho(t)).
\]
Next we observe that 
\[
\ti{\mu}^t - m_0^t = (\id-C_{w^t})^{-1}C_{w^t} m_0^t
\in L^2_s(\Sigma)
\] 
implying
\[
\| \ti{\mu}^t - m_0^t \|_{L^2_s}  = O(\rho(t)) \quad\text{and}\quad \ti{c}_0^t = O(\rho(t))
\]
since $\|m_0^t\|_\infty = O(1)$ (note $\ti{\mu}_0^t = \mu_0^t = m_0^t$).
Consequently $c_0^t = O(\rho(t))$ and thus $m^t(k) - m_0^t(k) = O(\rho(t))$
uniformly in $k$ as long as it stays a positive distance away from $\Sigma \cup \{\pm\I\kappa\}$.
\end{proof}

\bigskip

\noindent{\bf Acknowledgments.} We want to thank Ira Egorova and Helge Kr\"uger for helpful discussions
and the referees for valuable hints with respect to the literature.

\end{document}